\documentclass[11pt,reqno]{amsart}

\usepackage{amsmath, amsthm, amsopn, amssymb, graphicx, enumerate, geometry, afterpage, caption, subcaption, color, textcomp, bbm}

\newtheorem{thm}{Theorem}[section]
\newtheorem{lemma}[thm]{Lemma}
\newtheorem{prop}[thm]{Proposition}

\newtheorem{cor}[thm]{Corollary}

\theoremstyle{definition}

\numberwithin{equation}{section}

\def\={\  =\  }
\def\<{ \  \leq \  }
\def\0{\textbf{0}}

\newcommand{\R}{\mathbb{R}}
\newcommand{\Z}{\mathbb{Z}}

\newcommand{\E}{\mathbb{E}}
\renewcommand{\Pr}{\mathbb{P}}

\newcommand{\1}{\mathbbm{1}}

\def\eps{\varepsilon}
\renewcommand\P{\mathbb{P}}

\newcommand{\extB}{\partial_{\text{v}}}

\DeclareMathOperator\zero{\mathbf{0}}

\DeclareMathOperator{\Var}{Var}

\def\={\  = \  }
\def\be{\begin{equation}}
\def\ee{\end{equation}}
\def\ba{\be \begin{aligned}}
\def\ea{\end{aligned} \ee}

\author{Michael Aizenman}
\address{Michael Aizenman\hfill\break
    Departments of Mathematics and Physics,
    Princeton University,
    Princeton, NJ 08544,
    USA.}
\email{aizenman@princeton.edu}

\author{Ron Peled}
\address{Ron Peled\hfill\break
    School of Mathematical Sciences,
    Tel Aviv University,
    Israel.}
\email{peledron@post.tau.ac.il}

\title[On the Decay of Correlations in the
         RFIM]
{On the Decay Rate of Correlations in \\
 the Random Field Ising Model}

\title[Power-law upper bound on the correlations in the $2D$ RFIM]
{A Power-law upper bound on the correlations\\  in the $2D$ random field Ising model}

\dedicatory{\hspace{6cm} \it  To the memory of Joe Imry (1939-2018); \\
\hspace{6cm} insightful about physics, music, and life. }

\date{\today}

\begin{document}
\begin{abstract}
As first asserted by Y. Imry and S-K Ma,  the famed discontinuity of the magnetization  as function of the magnetic field in the two dimensional Ising model is eliminated, for all temperatures, through the addition of quenched random  magnetic field of uniform variance, even if that is small.
This statement is quantified here by a power-law upper bound on the
decay rate of the effect of boundary conditions on the magnetization in finite systems, as function of the distance to  the boundary.  Unlike  exponential decay which is only proven for strong disorder or high temperature, the power-law upper bound is established here for all field strengths and at all temperatures, including zero, for the case of independent Gaussian random field.  Our analysis proceeds through a streamlined and quantified version of the Aizenman-Wehr proof of the Imry-Ma rounding effect.
 \end{abstract}
\maketitle
\vspace{-1cm}
\tableofcontents

\section{Introduction}

\subsection{The Imry-Ma phenomenon for the $2D$ RFIM  and its quantification} \mbox{} \\[-1ex]

A first-order phase transition is one associated with phase coexistence, in which an extensive system admits at least two thermal equilibrium states which differ in their bulk densities of an extensive quantity.   The thermodynamic manifestation of such a transition  is the discontinuity in the derivative of the  extensive system's free energy  with respect to one of the  coupling constants which affect the system's energy.   At zero temperature, this would correspond to the existence of two infinite-volume ground states which differ in the bulk average of a local  quantity.

In what is known as the Imry-Ma ~\cite{IM75} phenomenon, in two-dimensional systems any   first-order transition is  \emph{rounded off}  upon the
introduction of arbitrarily weak  static, or \emph{quenched},  disorder in the parameter which is conjugate to the corresponding extensive quantity.

Our  goal here is to present  quantitive estimates of this effect, strengthening the previously proven infinite-volume statement~\cite{AW90} by:  i)  upper bounds on the dependence of the local density on a finite-volume's  boundary conditions, and ii) related bounds on the correlations among local quenched expectations, which are asymptotically independent functions of the quenched disorder.

The present discussion takes place in the context of the random-field Ising model.  In this case the original discontinuity is in the bulk magnetization, i.e. volume average of the local spin $\sigma_u$, and it occurs at zero magnetic field ($h=0$).  Since $h$ is the conjugate parameter to the magnetization, the relevant disorder for the Imry-Ma phenomenon is given by site-independent random field $(\eps \eta_u)$.
 More explicitly, the system consists of Ising spin variables $\{ \sigma_u\}_{u\in \Z^d}$, associated with the vertices of the $d$-dimensional lattice $\Z^d$,   with the Hamiltonian
\be \label{H}
  H(\sigma) := - \sum_{\substack{\{u,v\}\subseteq\Z^d} }
J_{u,v}\,  \sigma_u \sigma_v  - \sum_{v\in\Z^d} (h+ \eps\,\eta_v) \sigma_v \,,
\ee
and ferromagnetic translation-invariant coupling constants $\mathcal J = \{J_{u,v}\}$ ($J_{u,v} = J_{v,u} = J_{u-v,\zero}\geq 0$).

For convenience we  focus on the case that the $(\eta_v)$ are independent standard Gaussians.  However it is expected, and for many of the key results proven true, that the model's essential features are similar among all independent, identically distributed $(\eta_v)$ whose common distribution has a continuous component. \\

The main result presented here is the proof that in the two-dimensional case  at any temperature $T\ge 0$,  the effect on the local  quenched magnetization of the boundary conditions at distance $L $ away decays  by at least a power law ($1/L^\gamma$).   This may be viewed as a quantitative extension of the   uniqueness of the Gibbs state theorem~\cite{AW89, AW90}.  It also implies a similar bound on correlations within the infinite-volume Gibbs state.
A weaker  upper bound, at the rate  $1/\sqrt{\log \log L}$,   was recently presented in \cite{C17}, derived there by other means. \\

More explicitly:  as the first question it is natural to ask whether the addition of random field terms in the Hamiltonian ~\eqref{H} changes the Ising model's phase diagram, whose salient feature is the phase transition which for $d>1$ occurs at $h=0$   and low enough temperatures, $T< T_c$.  The initial prediction of Y. Imry and S-K Ma \cite{IM75} was challenged by other arguments, however it was eventually proven to be true: For $d\ge 3$ the RFIM continues to have a first-order phase transition at $h=0$~\cite{Imb85,BK88}, whereas in two dimensions at any $\eps \neq 0$ the model's bulk mean magnetization has a unique value for each $h$, and by implication it varies continuously in $h$ at any temperature, including $T=0$ ~\cite{AW89, AW90}.
Through the FKG property~\cite{FKG71} of the RFIM one may also deduce that in two dimensions, at any temperature $T\ge 0$ and
for almost every realization of the random field $\eta = (\eta_v)_{v\in \Z^2}$, the system has a unique Gibbs state. For $T=0$ this translates into uniqueness of the infinite-volume ground state configuration,   i.e.  configuration(s)  for which no flip of a finite number of spins results in lower energy.
Additional background and  pedagogical  review of the RFIM may be found in \cite[Chapter 7]{B06}.

Seeking  quantitative refinements of the above statement, we consider here the dependence
of the \emph{finite-volume} quenched magnetization $\langle\sigma_v\rangle^{\Lambda, \tau}$
on the boundary conditions $\tau$ placed on the exterior of a domain $\Lambda$.
We denote by $\langle-\rangle^{ \Lambda,\tau}$  the finite volume ``$\tau$ state'' quenched thermal average and by $\E{}$ the further average over the random field (both defined explicitly in Section~\ref{sec:Gibbs states}).
Due to the model's FKG monotonicity property the finite volume Gibbs states at arbitrary boundary conditions are bracketed between the $+$ and the $-$ state.
Hence the relevant order parameter is
\begin{equation} \label{eq:m}
  \begin{split}
  m(L) \equiv         m(L;T, \mathcal J, h,  \epsilon)    \,  := \, \frac 1 2 \left[
    \E[\langle\sigma_\0\rangle^{\Lambda(L), +} ] \ - \      \E[\langle\sigma_\0\rangle^{\Lambda(L), -}]
    \right] \,
  \end{split}
\end{equation}
where
\begin{equation}
  \Lambda_u(L):=\{v\in\Z^2\,\colon\, d(u,v)\le L\}\,\quad \mbox{,}\quad \  \Lambda(L) =   \Lambda_\textbf{0}(L) \,,
\end{equation}
with $d(u,v)$  the graph distance on  $\Z^2$ and $\0:=(0,0)$.

\begin{thm}\label{thm:power_law_bound}
In the two-dimensional random-field Ising model with a finite-range interaction $\mathcal J$ and
independent standard Gaussian random field  $(\eta_v)$, for any temperature  $T\geq 0$,  uniform field $h\in \R$,    and field intensity  $\eps>0$ there exist $C=C(\mathcal J, T, \eps) >0$ and $\gamma = \gamma(\mathcal J, T, \eps)>0$ such that for all large enough $L$
  \begin{equation} \label{eq:power_law}
m(L; T , \mathcal J, h, \epsilon) \le \frac{C}{L^{\gamma}} \, .
  \end{equation}
\end{thm}

For the nearest-neighbor interaction
\be \label{J_nn}
 J_{u,v} = J \, \delta_{d(u,v),1}
\ee
 the proof yields
\be     \gamma := 2^{-10}\cdot  \chi\left(\frac{50   J }{\eps}\right)\label{eq:gamma def}
\ee
in terms of the tail of the Gaussian  distribution function:
\begin{equation}\label{eq:Phi_def}
 \chi(t):=2\int_t^\infty \phi(s) \, ds \,  ,\qquad  \phi(s) = \frac{1}{\sqrt{2\pi}}e^{-s^2/2}\,  .
\end{equation}

The phenomenon and the arguments discussed in the proof are somewhat simpler to present in the limit of zero temperature, where the quenched random field is the only source of disorder.
We therefore start by proving Theorem~\ref{thm:power_law_bound} for this case, emphasizing the setting of nearest-neighbor interaction.
Then, in Section\,\ref{sec:positive_temperature} we present the changes by which the argument extends to  $T>0$.   With minor adjustments of the constants, discussed in Section~\ref{sec:rangeR},  the natural extension of the statement  to   translation-invariant pair interactions of finite range is also valid.

\subsection{Direct implications} \mbox{} \\[-1ex]

By the FKG inequality (see Section~\ref{sec:monotonicity}), the difference whose mean is the order parameter is non-negative for any $\eta$ (and all $T\geq 0$, $h\in \R$),
\be \label{eq:FKG11}
\langle\sigma_\0\rangle^{\Lambda(L), +}  - \langle\sigma_\0\rangle^{\Lambda(L), -}  \ \geq \ 0 \,.
\ee
Hence the bound on the mean \eqref{eq:power_law}
implies (through  Markov's inequality)  that this quantifier of sensitivity to boundary condition is similarly small with high probability.

The order parameter $m(L)$ controls also the covariances of: {\it i)}  the spins under  the infinite-volume quenched Gibbs states $\langle -\rangle \equiv \langle -\rangle_{T,\mathcal J, h, \eps \eta}$, and {\it ii)} of  the infinite-volume quenched Gibbs state magnetization $\langle \sigma_u\rangle$ under the random field fluctuations, over which the average is denoted by $\E(-)$.
To express these statements we denote
\ba\label{eq:truncated_correlations}
\langle \sigma_u;\sigma_v\rangle &:= \ \langle \sigma_u \sigma_v\rangle - \langle \sigma_u\rangle \, \langle \sigma_v\rangle  \\[1ex]
 \E(\langle \sigma_u \rangle; \langle \sigma_v\rangle)
  & :=  \E\Big(\langle \sigma_u \rangle \, \langle \sigma_v\rangle   \Big) -  \E(\langle \sigma_u \rangle)    \, \,   \E(\langle \sigma_v \rangle)
 \\
 &\,\,= \E\Big(\big[\langle \sigma_u \rangle - \E(\langle \sigma_0\rangle) \big]\,
 \big[\langle \sigma_v \rangle - \E(\langle \sigma_0\rangle) \big] \Big)
 \,.
\ea
Each of these truncated correlations is non-negative: in the former case due to the FKG property of the RFIM,  and in the latter due to monotonicity of  $\langle \sigma_u \rangle$ in $\eta$ and the Harris/FKG inequality for product measures.

As we  prove below (Lemma~\ref{lem:cov}), for pairs $\{u,v\}\in \Z^2$, if $d(u,v)>\ell$ then
\be \label{Es1}
 \E{(\langle \sigma_u;\sigma_v\rangle )} \  \leq   \    2\, m(\ell ;T,  \mathcal J, h, \epsilon)
\ee
while if $ d(u,v) \geq 2 \ell + R(\mathcal J)$, with  $R(\mathcal J) := \max\{d(u,v) \, : \, J_{u,v} \neq 0\}$ (the interaction's range) then
\be
  \label{eq:mcorr}
   \E(\langle \sigma_u \rangle; \langle \sigma_v\rangle) \ \leq \    4\,  m(\ell ;T,  \mathcal J, h, \epsilon) \,.
\ee
The comment made above in relation to \eqref{eq:m},  applies also here:   The non-negativity of $\langle \sigma_u;\sigma_v\rangle$, together with \eqref{Es1}, implies that with high probability it does not exceed $m(\ell ;T,  \mathcal J, h, \epsilon) $ by a large multiple.
The proof of \eqref{Es1} and \eqref{eq:mcorr} does not require the analysis which is developed in  this paper.  It is therefore postponed to  Section\,\ref{app:corr}.

 For  \eqref{eq:mcorr}
of particular interest is  $h=0$ and $T=0$.  In this case $\langle\sigma_u\rangle $ coincides with the infinite-volume ground state configuration $\widehat\sigma_u(\eta) $ which, as is already known, is unique for almost all $\eta$.   By the spin-flip symmetry  $\E(\widehat\sigma_u ) =0$, and the bound \eqref{eq:mcorr} translates into:
\be
   0\le \E(\widehat\sigma_u \widehat\sigma_v) \ \leq \
    4\,m(\ell ;0,  \mathcal J, 0, \epsilon)\,.
\ee

\subsection{A remaining question} \mbox{} \label{sec:remaining_question}\\[-1ex]

As we shall discuss in greater detail in Appendix~\ref{sec:high disorder}, at high enough disorder, i.e. large enough $\eps$,  the order parameter $m(L)$ decays exponentially fast in $L$.   Our results do not resolve the question of whether the two-dimensional model exhibits a disorder-driven  phase transition,  at which the decay rate changes from exponential to a power law, as the disorder is lowered (possibly even at $T=0$).    This remains among the interesting open problems concerning the Imry-Ma phenomenon in two dimensions, on which more is said in the open problem Section~\ref{sec:OP}.

\section{Gibbs equilibrium states}\label{sec:Gibbs states}

 \subsection{The Gibbs measure} \mbox{} \\[-1ex]

Discussing the RFIM on $\Z^2$
  we shall use the following terminology. Two vertices are deemed adjacent,  $u\sim v$, if they differ by  a unit vector. The graph distance on $\Z^2$ is denoted $d(u,v)$ and the graph ball of radius $L$ around $u$ is denoted $\Lambda_u(L)$, with $\Lambda(L)$ standing for $\Lambda_\0(L)$, as before Theorem~\ref{thm:power_law_bound}.
The \emph{edge boundary} of a subset $\Lambda \subset \Z^2$ (which is used in decoupling estimates) is denoted
\begin{equation}
  \partial_{\text{e}}\Lambda := \{(u,v)\,\colon\, u\in\Lambda,\, v\in\Z^2\setminus\Lambda,\,  J_{u,v} \neq 0 \}
\end{equation}
and the \emph{external boundary} (which is used when imposing boundary conditions) is
\begin{equation}
  \extB\Lambda:=\{v\in\Z^2\setminus\Lambda\,\colon\,\exists u\in\Lambda, J_{u,v} \neq 0 \} \,.
\end{equation}

\begin{figure}[h]
 \begin{center}
\includegraphics[width=.3 \textwidth]{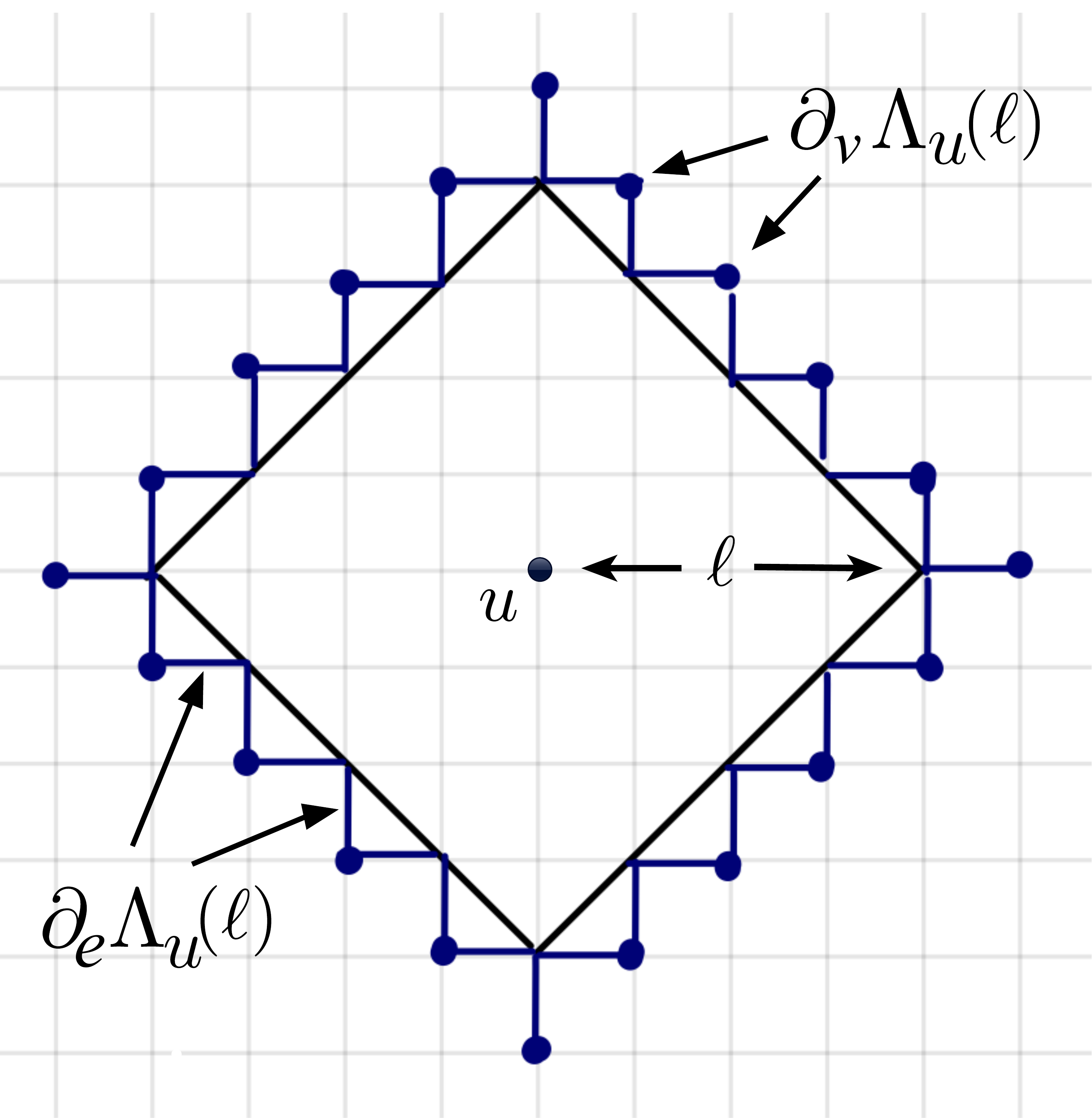}
\end{center}
\caption{A subset of $\Z^2$ of the form of $\Lambda_u(\ell)$ and its two boundary sets: the edge boundary $\partial_{\text{e}} \Lambda_u$ and the vertex (external)  boundary $\extB \Lambda_u$, both drawn for the case of the nearest-neighbor interaction.  }  \label{fig_boundary}
\end{figure}

The RFIM  Gibbs equilibrium state in the finite subset $\Lambda \subset \Z^2$,  at specified values of the parameters $(T, \mathcal J, h, \eps) $,   the random field $\eta$, and a configuration of boundary spin values $\tau:\extB\Lambda\to\{-1,1\}$,
 is  the probability measure  over $\Omega_\Lambda = \{-1,1\}^\Lambda$ given by
\begin{equation}\label{eq:P_Lambda_tau_def}
  \P^{\Lambda, \tau}(\sigma) := \frac{1}{Z^{\Lambda, \tau}}e^{-\frac 1T H^{\Lambda,\tau}(\sigma)},
\end{equation}
where
\begin{equation}\label{eq:Hamiltonian_def}
  H^{\Lambda,\tau}(\sigma) := -\sum_{u,v\in\Lambda} J_{u,v} \sigma_u \sigma_v - \sum_{(u,v)\in\partial_{\text{e}}\Lambda} J_{u,v} \sigma_u \tau_v -\sum_{v\in\Lambda}(h+  \eps \eta_v) \sigma_v\,
\end{equation}
and $Z^{\Lambda,\tau}$ is the corresponding normalizing factor (the ``partition function''). The associated expectation operator is denoted $\langle-\rangle^{ \Lambda,\tau}$. The notation $\P^{\Lambda, \pm}$ or $\langle-\rangle^{ \Lambda,\pm}$ indicates that $\tau$ is the corresponding uniform configuration $\tau \equiv +1$ or $\tau \equiv -1$. The notation $\P{}$ and $\E{}$ is used for the probability and expectation operators, respectively, of the further average over the random field.   \\

At $T=0$,  the measure $\P^{\Lambda, \tau}$ is supported on the almost-surely unique configuration which minimizes $H^{\Lambda,\tau}$.  These ground-state configurations, which depend on $\eps \eta$ and $(\mathcal J,h)$, are denoted here by $\sigma^{\Lambda, \tau} = (\sigma^{\Lambda, \tau}_v)_{v \in \Lambda }$ (The dependence on $\eta$ is not displayed, but it is in the focus of the discussion.)   \\

\subsection{Monotonicity properties}\mbox{} \label{sec:monotonicity}\\[-1ex]

In our  discussion we shall take advantage of  the known monotonicity property of the ferromagnetic Ising model,  which is that its Gibbs equilibrium states as well as the ground-state configurations, at given $\mathcal J$, $h$ and $\eps$, are increasing functions of the local field variables $\eta$ and of the boundary spin configuration $\tau$.
The statement  is a known  consequence of the FKG inequality \cite{FKG71}.  The $T=0$ version can also be seen through a more direct argument.  \\

Thus, for any region $\Lambda$ and pairs of
boundary conditions $\tau^-, \tau^+ :\extB\Lambda\to\{-1,1\}$:
\begin{equation}\label{eq:stochastic_domination_pos_temp}
\text{$\tau^-\le \tau^+ \quad \Longrightarrow \quad \P^{\Lambda,\tau^+}$ stochastically dominates $\P^{\Lambda,\tau^-}$}
\end{equation}
where an inequality between configurations is to be interpreted as holding pointwise.    (Unlike $\R$, the configuration space is only partially ordered, but that suffices for our purpose.) The following special case is noted for later reference
\begin{equation}\label{eq:FKG_pos_temp}
\text{$\tau^-\le \tau^+ \quad \Longrightarrow \quad  \langle\sigma_v\rangle^{\Lambda,\tau^-}\le \langle\sigma_v\rangle^{\Lambda,\tau^+}$ for each $v\in\Lambda$}.
\end{equation}
By related reasoning, the Gibbs state at $+$ (or $-$) boundary conditions  is stochastically  decreasing  (and correspondingly increasing) in its dependence on $\Lambda$.  In particular, for each $v \in \Lambda_1 \subset \Lambda_2 \subset \Z^2$:
\begin{equation}\label{eq:FKG_consequence_pos_temp}
\text{$\langle\sigma_v\rangle^{\Lambda_1,+} \ge \langle\sigma_v\rangle^{\Lambda_2,+}\quad$ and $\quad\langle\sigma_v\rangle^{\Lambda_1,-} \le \langle\sigma_v\rangle^{\Lambda_2,-}$}.
\end{equation}
\smallskip
The above inequalities hold also at $T=0$, where $\sigma^{\Lambda,\tau}_v$ substitutes for $\langle\sigma_v\rangle^{\Lambda,\tau}$. It is convenient to note this explicitly for later reference:
\begin{align}
&\tau^-\le \tau^+ \quad \Longrightarrow \quad  \sigma^{\Lambda,\tau^-}\le \sigma^{\Lambda,\tau^+}\label{eq:FKG},\\
&\sigma^{\Lambda_1,+}_v \ge \sigma^{\Lambda_2,+}_v \quad \mbox{and} \quad \sigma^{\Lambda_1,-}_v \le \sigma^{\Lambda_2,-}_v,\label{eq:FKG_consequence}\\
&\sigma^{\Lambda_1,+} - \sigma^{\Lambda_1,-}  \, \geq \,
\sigma^{\Lambda_2, +}  -  \sigma^{\Lambda_2, -}  \,  \geq \,  0\,,\label{eq:disag_mon}
\end{align}
with the second and third assertions holding for $v \in \Lambda_1 \subset \Lambda_2 \subset \Z^2$.

\section{Proof of the main result for $T=0$}

We start with the zero-temperature case of Theorem~\ref{thm:power_law_bound} as it already contains the main features of the problem while being technically simpler.    For a further simplification, we  consider first the nearest-neighbor interaction~\eqref{J_nn}.  The extension to finite-range interactions will follow in Section\,\ref{sec:rangeR}.

\subsection{Influence/disagreement  percolation}\mbox{} \\[-1ex]

Due to the monotonicity of the ground state in the boundary conditions, the order parameter $m(L)$ which is defined in \eqref{eq:m} can be viewed as the probability that the difference of the boundary conditions at distance $L$ from a site $v$ ``percolates'' to  $v$:
\begin{equation}\label{eq:p_L_def}
m(L) = m(L; 0, \mathcal  J, h, \epsilon)  \=  \P\left(\sigma^{\Lambda(L),+}_{\mathbf{0}}> \sigma^{\Lambda(L),-}_{\mathbf{0}}\right)  \,.
  \end{equation}

\noindent{\bf Remark:} Disagreement percolation provides a concrete manifestation of the influence of the boundary condition.  The terms \emph{disagreement percolation}   and
 \emph{influence percolation}
are  almost interchangeable: the former referring to  specific manifestations of the latter.  The term percolation is called for since the influence/disagreement spreads only along connected sets.

To learn about $m(L)$  we find it useful to consider the following functions of the disorder:
\be \label{eq:D_l_def}
  D_{\ell}(\eta)   \ := \   \sum_{v\in \Lambda(\ell)} \1[ \sigma^{\Lambda(3\ell),+}_v\neq \sigma^{\Lambda(3\ell),-}_v]\, ,
\ee
the number of sites in $\Lambda(\ell)$ to which the difference of the boundary conditions
 imposed on the boundary of $\Lambda(3\ell)$ has ``percolated'', and
\be \label{eq:B_l_def}
  B_{\ell}(\eta)   \ := \  \sum_{(u,v)\in \partial_{\text{e}} \Lambda(2\ell)} J_{u,v} \, \1[\{\sigma^{\Lambda(3\ell)\backslash\Lambda(\ell),+}_u\neq \sigma^{\Lambda(3\ell)\backslash\Lambda(\ell),-}_u\}\cap\{\sigma^{\Lambda(3\ell)\backslash\Lambda(\ell),+}_v\neq \sigma^{\Lambda(3\ell)\backslash\Lambda(\ell),-}_v\}]\,  .
\ee
The latter is the combined strength of the  edges crossing a separating surface at half the distance of $\Lambda(\ell)$ to the boundary of $\Lambda(3\ell)$, which contribute to the surface tension.

\subsection{The surface tension}  \label{sec:T}  \mbox{} \\[-1ex]

One may learn about the probability distribution of the disagreement set $D_\ell$   through consideration of
the surface tension, which for scale $\ell$ (always a positive integer) is defined as
\begin{multline} \label{T_def}
\mathcal T_{\ell}(\eta)  :=  \\
 -\left[ \mathcal E^{+,+}(\Lambda(3\ell)\backslash \Lambda(\ell)) + \mathcal E^{-,-}(\Lambda(3\ell)\backslash  \Lambda(\ell)) -  \mathcal E^{+,-}(\Lambda(3\ell)\backslash \Lambda(\ell)) - \mathcal E^{-,+}(\Lambda(3\ell)
\backslash \Lambda(\ell)) \right] \,.
\end{multline}
Here $\mathcal E^{s,s'}(\Lambda(3\ell)\backslash \Lambda(\ell)) $ denotes the minimal value of the Hamiltonian
$H^{\Lambda(3\ell)\backslash \Lambda(\ell), \tilde{\tau}}$ (see~\eqref{eq:Hamiltonian_def}) over spin configurations satisfying the boundary conditions
\be\label{eq:tilde_tau_def}
\tilde{\tau}_v \, = \, \begin{cases} s & v\in \extB \Lambda(3\ell) \\
        s' & v\in \extB (\Z^2\setminus\Lambda(\ell))
        \end{cases} \, .
\ee

Our analysis proceeds by contrasting  a natural upper bound on the surface tension, with the analysis of the not-improbable fluctuations of $\mathcal T_\ell(\eta)$.   For the upper bound we have:

\begin{thm} \label{thm:T1}  In the RFIM with nearest-neighbor interaction, for each configuration of the random field:
\be \label{T_B}
 \mathcal T_{\ell}(\eta)
 \  \leq  \  4  \, B_{\ell} (\eta)  \  \leq  \  8 J \, |\extB \Lambda(2 \ell)|  \,.
\ee
\end{thm}
\begin{proof}
Let $A$ be the set of vertices in $\Lambda(3\ell)\backslash \Lambda(\ell)$ on which there is equality between the ground-state configurations with $++$ and $--$ boundary conditions. The monotonicity property \eqref{eq:FKG} implies that all ground states on $\Lambda(3\ell)\backslash \Lambda(\ell)$ must coincide on $A$. Consider making two modifications to the Hamiltonian in the domain $\Lambda(3\ell)\backslash \Lambda(\ell)$: First, rigidly restrict the spin values at all vertices in $A$ to their common value in these ground states. This clearly has no effect on the energies of the ground-state configurations considered above. Second, remove the energy terms corresponding to bonds in $\partial_{\text{e}}\Lambda(2\ell)$ whose endpoints do not intersect $A$. This change may affect the energy of each of the four ground states by at most $ B_{\ell}(\eta)$. Once both changes are made, the Hamiltonian decomposes into a sum of two  terms, in whose minimization there is no interaction between the effects of the two components of the boundary. Thus the surface tension based on the modified Hamiltonian vanishes.

It follows that $\mathcal T_\ell(\eta) \ \leq 4\,    B_\ell(\eta) $
as claimed in the first inequality in \eqref{T_B}.  The second is its elementary consequence.
\end{proof}

The upper bound  which \eqref{T_B} yields on  $\mathcal T_\ell$ will be contrasted with the implications of the following representation.

\begin{thm} \label{thm:T2} For the RFIM with IID Gaussian random fields,
the surface tension bears  the following relation with disagreement percolation:
\be \label{DD2}
 \mathcal T_{\ell}(\eta)   \=  2 \eps\,
 \int_\R D_\ell(\eta^{(t)}) \, dt \ = \
 \ \frac{2 \eps  }{\sqrt{|\Lambda(\ell)|}}\, \E _{\widehat \eta } \left(  \frac{D_{\ell}(\eta) }{ \phi(\widehat \eta) } \right)  \, ,
\ee
where: \begin{enumerate}[1)]
\item
 $\eta^{(t)}$ is defined by adding a uniform field of intensity $t$ in $\Lambda(\ell)$,
\begin{equation} \label{def:t}
  \eta^{(t)}_v := \begin{cases}
    \eta_v + t&v\in\Lambda(\ell)\\
    \eta_v&\text{otherwise}
  \end{cases} \, .
\end{equation}
\item the variable $\widehat \eta $ is defined as
\be \label{eq:etahat}
\widehat \eta := \frac 1{\sqrt{| \Lambda (\ell)|}} \sum _{v\in  \Lambda (\ell)} \eta_v   \,  .
\ee
\item  $\E _{\widehat \eta } $ represents an average over  $\widehat \eta$ at fixed values of the other, orthogonal, Gaussian degrees of freedom which determine $\eta$.
 \item  $\phi$ is the Gaussian density function \eqref{eq:Phi_def}.

 \end{enumerate}
\end{thm}

(An alternative presentation of $\E_{\widehat \eta}$:  decomposing $\eta$ as a sum of two independent Gaussian fields $\eta_1, \eta_2$ with $\eta_1 \equiv \frac 1{\sqrt{| \Lambda (\ell)|}}\widehat \eta$ on $\Lambda(\ell)$, and $\eta_1 \equiv 0$ outside $\Lambda(\ell)$, the operation $\E_{\widehat \eta }$ represents conditional expectation, given $\eta_2$.)

\bigskip

\begin{proof}

To derive \eqref{DD2} we approach $\mathcal T_\ell(\eta) $ through another function, $G_\ell(\eta)$, which has already played a key role in the proof of the absence of symmetry breaking in the two-dimensional RFIM~\cite{AW89, AW90}.   Its zero-temperature version corresponds to the difference in the ground-state energies in $\Lambda(3\ell)$ between the $+$ and $-$ boundary conditions:
\begin{equation}\label{eq:G_def}
G_\ell(\eta) \, := \,
  -\left[\mathcal E^{+}(\Lambda(3\ell)) -  \mathcal E^{-}(\Lambda(3\ell))\right]
\end{equation}
with $\mathcal E^{\pm}(\Lambda(3\ell)):=H^{\Lambda(3\ell), \pm}(\sigma^{\Lambda(3\ell),\pm})$.\\

The two functions are linked by the relation
\be \label{eq:TandG}
\mathcal T_\ell(\eta) \ = \
  \lim_{t\to\infty}G_\ell (\eta^{(t)}) - G_\ell (\eta^{(-t)})
\ee
with $\eta^{(t)}$ defined by adding a uniform field of intensity $t$ in $\Lambda(\ell)$, as described in \eqref{def:t}.
Equality \eqref{eq:TandG} is based on the observation that if  $|h+\eps \eta_v| >  4J$
then $\sigma^{\Lambda(3\ell),\pm}_v $ are both given by ${\rm{sign}}(h+\eps\eta_v)$
 ($4$ appears here as the number of neighbors of $v$ in $\Z^2$).\\

The function
 $G_\ell(\eta)$ is Lipschitz continuous and non-decreasing in each of the coordinates of $(\eta_v)$, $v\in\Lambda(3\ell)$, with
\begin{equation}\label{eq:G_derivative_formula}
  \frac{\partial}{\partial\eta_v} G_\ell(\eta) = \eps
  \left[ \sigma^{\Lambda(3\ell),+}_v(\eta) - \sigma^{\Lambda(3\ell),-}_v(\eta) \right]  \, =\,  2 \eps\, \1_{\sigma^{\Lambda(3\ell),+}_v(\eta) \neq \sigma^{\Lambda(3\ell),-}_v(\eta) }
\end{equation}
for Lebesgue-almost-every $\eta$.   Combining this with \eqref{eq:TandG} one gets
\be\label{eq:surface_tension_via_derivative}
\begin{split}
\mathcal T_\ell(\eta) &= \eps\int_{-\infty}^\infty
\sum_{v \in \Lambda(\ell)} \left[ \sigma_v^{\Lambda(3\ell),+}(\eta^{(t)}) -
\sigma_v^{\Lambda(3\ell),-}(\eta^{(t)}) \right] \, dt\, = \, 2\eps
\int_{-\infty}^\infty D_{\ell}(\eta^{(t)})\, dt\, .
\end{split}
\ee
The shift by $t$ affects the random field's
 normalized sum over $\Lambda(\ell)$, which we denote by $\hat \eta = \sum_{v\in\Lambda(\ell)} \eta_v/  \sqrt{|\Lambda(\ell)|} $ but it  does not affect  the independently distributed  degrees of freedom which as Gaussian variables are orthogonal  to it, $\hat \eta^{(\perp)}$.

Writing $\eta= (\hat \eta, \hat \eta^{(\perp)}) $ and
$t\sqrt{ | \Lambda(\ell)|} = s$
the change $\eta \mapsto \eta^{(t)}$ corresponds to the shift
$(\hat \eta, \hat \eta^{(\perp)})  \mapsto  (\hat \eta +s, \hat \eta^{(\perp)}) $.  Since the component $\hat\eta$ has the standard Gaussian distribution, of density $\phi(\hat \eta)$, the above integral  can be rewritten as:
\begin{equation}\label{T_B2}
\begin{split}
&\int_{-\infty}^\infty D_{\ell}(\eta^{(t)})\, dt  =
\frac{1}{\sqrt{ | \Lambda(\ell)|} } \int_{-\infty}^\infty  D_{\ell}((\hat \eta +s , \hat \eta^{(\perp)})  \, d s
   = \frac{1}{\sqrt{ | \Lambda(\ell)|} } \int_{-\infty}^\infty  D_{\ell}((s , \hat \eta^{(\perp)})  \, d s    \\
 & = \frac{1}{\sqrt{ | \Lambda(\ell)|} } \int_{-\infty}^\infty  D_{\ell}((s , \hat \eta^{(\perp)})\phi( s)^{-1}\cdot\phi( s)  \, d s = \frac1 {\sqrt{ | \Lambda(\ell)|}}\, \E _{\widehat \eta } \left(  D_{\ell}(\eta) \,
\phi( \hat \eta)^{-1}\right)\, .
\end{split}
\end{equation}
\end{proof}

\smallskip

\subsection{Proof outline  for the RFIM ground states} \mbox{} \\[-1ex]

Influence percolation quantities  appear in both the surface tension formula \eqref{DD2} and the upper bound \eqref{T_B}.  The combination of these two  yields   the following relation, which underlies our analysis:

\begin{eqnarray}  \label{101}
 \frac{2 \,   \E(B_\ell(\eta))}{\eps \sqrt{|\Lambda(\ell)|}}
 & \geq &   \E \left(  \frac{D_{\ell}(\eta) }{| \Lambda(\ell)|}\,
 \frac{1 }{ \phi(\widehat \eta) }  \right) \,.
\end{eqnarray}

To motivate the direction which the discussion is about to take,
let us  note that \eqref{101}
allows a streamlined proof of the following  statement, which is among the significant results established in \cite{AW90}.

\begin{cor}
In the  two-dimensional RFIM with Gaussian random field, for any $\eps \neq 0$, the system has a unique ground-state configuration.
\end{cor}
\begin{proof}
The monotonicity relations~\eqref{eq:FKG_consequence} imply that as the domains $\Lambda_n$ increase to $\Z^2$, the ground state $\sigma^{\Lambda_n, +}$ converges pointwise to a limiting ground state $\sigma^+$, which is, moreover, independent of the choice of exhausting sequence $\Lambda_n$. the ground state $\sigma^-$ is defined similarly with $-$ boundary conditions. The monotonicity relation~\eqref{eq:FKG} then shows that uniqueness of the ground state is equivalent to the vanishing of the quantity
\be   m(\infty) := \lim_{\ell \to \infty}   m(\ell)  \= \P\left(\sigma^{+}_v\neq \sigma^{-}_v\right) \, ,
\ee
where  $v$ is an arbitrary point in $\Z^2$.

The monotonicity relation~\eqref{eq:disag_mon} further allows to deduce from \eqref{101} that
\be
\frac{C \,J}{\eps }
    \geq \
  \E \left(  \left[ \frac{ 1}{ | \Lambda(\ell)|} \sum_{v\in \Lambda(\ell)}  \1[ \sigma^{+}_v\neq \sigma^{-}_v]
\right]
\,  \frac{1 }{ \phi(\widehat \eta) }  \right)  \,   ,
\ee
where $C>0$ is an absolute constant.

The pair of ground states $(\sigma^+, \sigma^-)$ form an ergodic process under translations (as a factor of the IID process $\eta$). This
allows to conclude that in the limit $\ell \to \infty$ the quantity
 $ \frac{ 1}{ | \Lambda(\ell)|} \sum_{v\in \Lambda(\ell)}  \1[ \sigma^{+}_v\neq \sigma^{-}_v] $
converges almost surely to its mean, which is $m(\infty)$.
Hence, using Fatou's lemma (for the second inequality)

\begin{equation} \label{contra}
\begin{split}
\frac{\rm{C} \, J}{\eps \,  }  &\ \geq\ \lim_{\ell \to \infty} \E \left(  \left[ \frac{ 1 }{ | \Lambda(\ell)|} \sum_{v\in \Lambda(\ell)}  \1[ \sigma^{+}_v\neq \sigma^{-}_v]
\right]
\,  \frac{1 }{ \phi(\widehat \eta) }   \right)  \ \geq \
 \, \, \E \left(  \frac{m(\infty) }{ \phi(\widehat \eta) }   \right) \ =  \\
 &\ =\   m(\infty)\, \int_{-\infty} ^ \infty 1\, dx \= m(\infty) \cdot \infty \, .
\end{split}
\end{equation}
This can hold true only if  $m(\infty) =0$.
\end{proof}

The ergodicity argument  is of not much help for the finite-volume bounds which are sought here.  It may however be substituted  by  more quantitative estimates, which are derived below under  the assumption that $m(\ell) \to 0$ at only  a sub-power slow rate.
To produce a contradiction which replaces \eqref{contra} we shall first show that  \eqref{101}  implies the following
anti-concentration bound.
\begin{prop}\label{prop:quantitative_lower_bound}
  For each integer $\ell\ge1$,
  \begin{equation}  \label{202}
  \P\left(\frac{D_\ell}{\E(D_\ell)} < \frac{1}{2}\right)\,\ge\, \chi\left(\frac{4J}{\eps}\cdot\frac{|\extB\Lambda(2\ell)|}{\sqrt{|\Lambda(\ell)|}}\cdot \frac{m(\ell - 1)}{m(4\ell)} \right)\,,
\end{equation}
where $\chi$ is the standard Gaussian distribution's two-sided tail \eqref{eq:Phi_def}.
\end{prop}

This  bound \eqref{202} will be contrasted  with a conditional concentration-of-measure estimate, derived through the following two steps.   For the convenience of presentation we summarize here the key statements, and postpone their proofs to the sections which follow.  \\

I)  Slow decay of a monotone sequence implies the existence of long stretches of somewhat comparable values:

\begin{prop}\label{prop:comp_decay}
For any monotone non-increasing sequence $(p_j)$ satisfying $0\le p_j\le 1$,  and any $\alpha>0$:  if for some  $k\ge 1$ it holds that
\be
p_k\ge k^{-\alpha}
\ee
  then there exists an integer~$n$ in the range $\sqrt{k}\le n\le k$ such that for all $1\le j \leq n$,
  \begin{equation}\label{eq:comp_dec}
    p_{n}\le p_j\le p_{n}\left(\frac{n}{j}\right)^{2\alpha} \,.
  \end{equation}
\end{prop}
The proposition will be employed with $(m(j))$ as the sequence $(p_j)$.

\mbox{} \\

II)  A conditional variance bound:

\begin{prop}\label{prop:var_bound}
  For each $0<\alpha\le\frac{1}{4}$ there exists $L_0>0$ such that the following holds for all integer $L\ge L_0$. If
  \be
  m(L)\ge L^{-2\alpha}
  \ee  and
  \begin{equation}\label{eq:assumed_decay}
    m(L) \le m(j)\le m(L) \left(\frac{L}{j}\right)^{2\alpha},\quad 1\le j\le L
  \end{equation}
  then
  \begin{equation} \label{eq:Var_bnd}
    \Var\big(D_{  \lfloor L/4\rfloor }\big) \  \le \  241 \cdot \alpha\cdot \big(\E \left( D_{ \lfloor L/4\rfloor}\right) \big)^2.
  \end{equation}
  \end{prop}

\mbox{ } \\

Combining Proposition~\ref{prop:comp_decay} and Proposition~\ref{prop:var_bound} with the assumption of sub-power decay of $(m(j))$ shows the existence of an infinite sequence of $L$s for which \eqref{eq:assumed_decay} and \eqref{eq:Var_bnd} hold. With $\ell=\lfloor L/4\rfloor$, Chebyshev's inequality and~ \eqref{eq:Var_bnd} imply that along this sequence
\begin{equation} \label{303}
    \P\left(\frac{D_\ell}{\E \big( D_\ell \big)} < \frac{1}{2}\right) \,
   \le \, 1000\alpha.
\end{equation}
At the same time, for $\alpha \to 0$, the ratio $m(\ell-1)/m(4\ell)$ tends to $1$ by \eqref{eq:assumed_decay}, and the argument of $\chi$ in \eqref{202} is bounded by $\rm{Const.} J/\eps$, uniformly in $\ell$ and $\alpha$.
Hence, for  small enough $\alpha>0$, \eqref{303}   is in contradiction with the anti-concentration bound \eqref{202}.   \\

The above line of reasoning  allows  to conclude that the initial assumption of sub-power decay is false.  A quantitative version of the argument, proving   the zero-temperature case of Theorem~\ref{thm:power_law_bound},  is presented in Section~\ref{sec:nailing_T_0} after the derivation of the above three propositions.

\bigskip

\subsection{The anti-concentration estimate} \mbox{} \\[-1ex]

In the proof of Proposition~\ref{prop:quantitative_lower_bound} we shall make use of the following variational principle.

\begin{lemma}\label{lem:min_integral}
Let  $w: \R \mapsto [0,\infty) $ be a symmetric ($w(-x) = w(x)$), non-increasing in $|x|$, probability density function on $\R$, i.e. satisfying $\int_\R w(x) dx =1$.
Then, for any
$ p \in (0, 1]$,
  \be \label{eq:min}
      \min \Big\{\int_\R f(x)\,  dx \,  \quad {\Big |}  \, \,  0\le f \le 1\,,\,
      \int_\R f(x) \, w(x) \, dx \= 1- p {\Big\}} \,
    = \, 2 \, q
\ee
where the variation is over measurable functions satisfying the stated conditions, and   $q$ is the unique value related to $p$ by
 \be \label{eq:W_q_p}
\int_{|x|> q} w(x) \, dx\=   p.
\ee
\end{lemma}
\begin{proof}
For each test function satisfying the conditions in \eqref{eq:min},
\begin{eqnarray}  \label{eq:aux_lemma}
\int_\R f(x) dx &=&  2 q  \ -  \  \int_{|x| \le q} [1-f(x) ]\, dx  \ + \ \int_{|x| > q} f(x) \, dx \notag \\
&\geq &  2 q  \ -  \ \frac{1}{w(q)} \left[  \int_{|x| \le q} [1-f(x) ] w(x) \, dx  \ - \ \int_{|x| > q} f(x) \, w(x) \, dx \right] \notag \\
&= &  2 q  \ +  \ \frac{1}{w(q)} \left[  \int_{|x| \le q}  w(x) \, dx  \ - \ \int_\R  f(x) w(x) \, dx \right] \ = \ 2q.
\end{eqnarray}
Equality in \eqref{eq:min} is attained for the indicator function $f(x) = \1_{[-q,q]}(x)$.
\end{proof}
\medskip

\noindent {\bf Remarks:}   1) For a structural grasp of Lemma~\ref{lem:min_integral} one may note that by a rearrangement argument it suffices to restrict the variation there to $f$ which are also symmetric and non-increasing in $|x|$.   A convexity argument allows to further restrict to extreme points in the convex set of admissible functions.  These are functions satisfying the constraints  but taking (almost everywhere) only the values $0$ and $1$.  These two conditions single out the indicator function $\1_{[-q,q]}(x)$ (or $\1_{(-q,q)}(x)$), and thereby imply that it is a minimizer for \eqref{eq:min}.  \\

2)
The assumptions of symmetry and monotonicity of the probability density $w(x)$  are not essential, and upon the natural reformulation of \eqref{eq:W_q_p} can be omitted. They are however satisfied by the Gaussian density function $\phi(x)=e^{-x^2/2}/\sqrt{2\pi}$.\\

The above will next be used to prove the stated estimate.

\begin{proof}[Proof of Proposition \ref{prop:quantitative_lower_bound}]
Let $A$  be the event $\big\{ \eta   :   D_\ell \geq     \E(D_\ell)/2 \, \big \} $  and let us denote its probability as $1-p$, i.e.
\be
 \Pr\Big( D_\ell \geq    \frac 1 2  \E(D_\ell) \Big) \=   \Pr(A) \=  1-p \,.
\ee
From   \eqref{101} one may deduce:

\begin{eqnarray}  \label{1012}
 \frac{2 \,}{\eps} \, \frac{ \E(B_\ell)}{ \sqrt{|\Lambda(\ell)|}}
\,  \frac{|\Lambda(\ell)|}{\E(D_\ell)}
 & \geq &   \frac 12\, \E   \left(  \1[A] \,
 \frac{1 }{ \phi(\widehat \eta) }  \right) \, .
\end{eqnarray}
Expressed in terms of the conditional probability of $A$, conditioned on $\widehat \eta$, the term on the right is, by Lemma~\ref{lem:min_integral},
\begin{eqnarray}
 \label{var_bound_2}
 \E \left(  \1[A] \,
 \frac{1 }{ \phi(\widehat \eta) }  \right) \    & = &
  \int_{-\infty} ^\infty \P \big( A \big| \, \widehat \eta = x \big) \, dx  \notag \\
  & \geq &
  \min \Big\{ \int_{-\infty} ^\infty  f(x) \, dx  \, \Big|  \,  0 \leq f\leq 1\, , \,
\int_\R  f(x) \, \phi(x) dx = 1-p \Big \} \notag \\[2ex]
& \geq &
\,  2 \, q
\end{eqnarray}
with $q$ defined by:
\be
\chi(q) \ \equiv \  \int_{|x|>q} \phi(x) \, dx \= p \,.
\ee

Combining \eqref{var_bound_2} with \eqref{1012} we learn that
\be
\frac{2 \,}{\eps} \, \frac{ \E(B_\ell)}{ \sqrt{|\Lambda(\ell)|}}
\,  \frac{|\Lambda(\ell)|}{\E(D_\ell)}
 \ \geq \  q \,.
\ee
Hence
\be
\chi\Big(\frac{2 \,}{\eps} \, \frac{ \E(B_\ell)}{ \sqrt{|\Lambda(\ell)|}}
\,  \frac{|\Lambda(\ell)|}{\E(D_\ell)} \Big)  \ \leq \ \chi(q) \= p \= \Pr\Big ( D_\ell <    \frac 12 \E(D_\ell)  \Big).
\ee
To obtain the conclusion \eqref{202} of the proposition, it remains to note that, by the definitions \eqref{eq:p_L_def}, \eqref{eq:D_l_def}, \eqref{eq:B_l_def} of $m(j)$, $D_\ell$ and $B_\ell$, together with the monotonicity inequality \eqref{eq:disag_mon},
\begin{align}
  \E(B_\ell) &\le 2 J \, |\extB\Lambda(2\ell)|\,m(\ell-1),\label{eq:B_ell_bound}\\
  \E(D_\ell) &\ge |\Lambda(\ell)|\,m(4\ell).\qedhere
\end{align}
 \end{proof}

 \bigskip

\subsection{Implications of slow decay}\mbox{} \\[-1ex]

We next show that slow decay of a monotone sequence implies the existence of long stretches of somewhat comparable values.

\begin{proof}[Proof of Proposition~\ref{prop:comp_decay}]
Assume that for some $k$ and $\alpha >0$
\be\label{eq:p_k_assumption}
p_k\ge k^{-\alpha} \,.
\ee

  As $(p_j)$ is non-increasing we need only prove the right-hand inequality in \eqref{eq:comp_dec}. Define a sequence $k =: k_0 > k_1 > \cdots > k_t$    inductively by letting $k_m$ be the maximal integer in $(0,k_{m-1})$ such that $p_{k_m}> p_{k_{m-1}}\left(\frac{k_{m-1}}{k_m}\right)^{2\alpha}$ provided such an integer exists, and denoting by $t$  the first value of $m$ beyond which the construction cannot proceed.  By construction,   for all $0< j\le k_t$: $p_j\le p_{k_t}\left(\frac{k_t}{j}\right)^{2\alpha}$.   If $t=0$ the claim follows  with $n := k$. Otherwise, using~\eqref{eq:p_k_assumption},
  \begin{equation}
    1\ge p_{k_t}>p_{k_{t-1}}\left(\frac{k_{t-1}}{k_t}\right)^{2\alpha}>p_{k_{t-2}}\left(\frac{k_{t-2}}{k_t}\right)^{2\alpha}>\cdots>p_k\left(\frac{k}{k_t}\right)^{2\alpha}\ge \left(\frac{\sqrt{k}}{k_t}\right)^{2\alpha}
  \end{equation}
  so that $k_t\ge \sqrt{k}$ and the claim \eqref{eq:comp_dec} holds true with $n:=k_t$.
 \end{proof}

\mbox{}

Next we turn to the implications of slow decay on the variance of the size of the disagreement set, $ \Var (D_\ell)$.

\begin{proof}[Proof of Proposition~\ref{prop:var_bound}]
Assume  that $m(L)\ge L^{-2\alpha}$, and that
\eqref{eq:assumed_decay} holds for all $1\le j\le L
$.    Throughout the proof we set
\begin{equation*}
  \ell := \lfloor L/4\rfloor.
\end{equation*}
For  $v\in \Lambda(\ell)$ let $E_v$ denote the event
$\{ \eta \, :\, \sigma^{\Lambda(3\ell),+}_{v}(\eta) \neq  \sigma^{\Lambda(3\ell),-}_v(\eta)\}$.  In this notation:
\begin{equation}\label{eq:var_as_sum_over_pairs}
  \Var\left(D_\ell \right) = \sum_{v,w\in \Lambda(\ell)}  \big[\P(E_v\cap E_w) - \P(E_v)\P(E_w)\big]\, .
\end{equation}
We proceed to bound the terms in this sum.  \\

By the FKG monotonicity \eqref{eq:disag_mon} and the definition~\eqref{eq:p_L_def} of $(m(j))$,  for any site
$v \in \Lambda(\ell)$,
\be \label{lower_E_v}
\P(E_v) \ \geq \   m(4\ell)\ge m(L)
 \ee
and for any pair $v, w \in \Lambda(\ell)$, $v\neq w$,
\be \label{eq:uv}
  \P(E_v\cap E_w) \< m(r(v,w))^2
\end{equation}
with
\begin{equation}\label{eq:r_v_w_def}
  r(v,w):=\lfloor (d(v,w)-1)/2\rfloor
\end{equation}
and $d(v,w)$ the distance between the two sites.   The bound~\eqref{eq:uv}
holds since if both $v$ and $w$ are affected by boundary conditions placed  outside of $\Lambda(3\ell)$ then each spin is necessarily affected also by boundary conditions placed at distance $r(v,w)$ from the site. However, these two events are independent, since they depend only on the random fields in a pair of disjoint neighborhoods of $v$ and $w$.

 For pairs at distance $d(v,w) \< 2$ we shall employ the simpler bound:
\begin{equation}\label{eq:E_v_E_w_simple_bound}
  \P(E_v \cap E_w) - \P(E_v)\P(E_w)\le \P(E_v) \le m(2\ell).
\end{equation}

Thus
under the assumption \eqref{eq:assumed_decay} we get
\begin{equation}\label{eq:variance_second_bound}
\begin{split}
  \Var\left( D_\ell \right)&\le |\Lambda(2)|\cdot|\Lambda(\ell)| m(2\ell) + \sum_{\substack{v,w\in \Lambda(\ell)\\d(v,w)\ge 3}} \left(m(r(v,w))^2 - m(L)^2\right)\\
   &\le |\Lambda(2)|\cdot|\Lambda(\ell)| m(2\ell) + m(L)^2\sum_{\substack{v,w\in \Lambda(\ell)\\d(v,w)\ge 3}} \left(\left(\frac{L}{r(v,w)}\right)^{4\alpha} - 1\right).
\end{split}
\end{equation}
The sum in the last bound can be estimated through the  observation that most pairs $v,w\in\Lambda(\ell)$ are at distance of order $\ell$, in which case $\frac{L}{r(v,w)}$ is of order $1$. As $\alpha$ is small,  for such pairs
$\big(\frac{L}{r(v,w)}\big)^{4\alpha} - 1$ is of order $\alpha$.  This leads to a bound of order $\alpha\, m(L)^2 L^4$ on the variance, which in light of
\eqref{lower_E_v}  is of the order $\alpha \Big(\E(D_\ell)\Big)^2$.

We proceed to make this argument precise. We first note that
\begin{equation}\label{eq:sum_to_integral}
\begin{split}
  &\sum_{\substack{v,w\in \Lambda(\ell)\\d(v,w)\ge 3}} \left(\left(\frac{L}{r(v,w)}\right)^{4\alpha} - 1\right) = \sum_{j=1}^{\ell}|\{(v,w)\subseteq \Lambda(\ell)\,:\,r(v,w)=j\}| \left(\left(\frac{L}{j}\right)^{4\alpha} - 1\right)\\
  &\le |\Lambda(\ell)| \sum_{j=1}^{\ell}32j \left(\left(\frac{L}{j}\right)^{4\alpha} - 1\right).
\end{split}
\end{equation}
For large $\ell$ and $0<\alpha\le \frac{1}{4}$
\begin{equation*}
\begin{split}
  \sum_{j=1}^{\ell}j \left(\left(\frac{L}{j}\right)^{4\alpha} - 1\right) & \leq \ \int_{1}^{\ell+1} L^{4\alpha}x^{1-4\alpha}dx - \int_0^{\ell} x dx \\
  &\leq \  \frac{\ell^2}{2}\left[\frac{2}{2-4\alpha}\cdot\left(\frac{L}{\ell+1}\right)^{4\alpha}\cdot\left(\frac{\ell+1}{\ell}\right)^2-1\right]\ \leq\   15\, \alpha \, \ell^2.
\end{split}
\end{equation*}
 Substituting this into \eqref{eq:variance_second_bound}, along with $|\Lambda(r)|\ge 2r^2$,  we conclude  that
\begin{equation}\label{eq:var bound with L}
\begin{split}
  \Var\left(D_\ell \right)&\ \le \  |\Lambda(2)|\cdot|\Lambda(\ell)|m(2\ell) + m(L)^2 \cdot |\Lambda(\ell)|\cdot32\cdot15\alpha \ell^2\\
   &\ \le\ |\Lambda(2)|\cdot|\Lambda(\ell)|m(2\ell) + 240 \alpha \left(m(L) |\Lambda(\ell)|\right)^2.
\end{split}
\end{equation}

It remains to observe that, by \eqref{lower_E_v},
\begin{equation*}
  \E \big( D_\ell \big) \ \ge \ m(L) |\Lambda(\ell)|.
\end{equation*}
Moreover, by our assumptions that $m(L)\ge L^{-2\alpha}$ and that \eqref{eq:assumed_decay} holds,
\begin{equation*}
  |\Lambda(2)|\cdot|\Lambda(\ell)|m(2\ell)\  \le\  |\Lambda(2)|\cdot|\Lambda(\ell)|m(L) \left(\frac{L}{2\ell}\right)^{2\alpha} \  \le \  \alpha\Big( m(L)|\Lambda(\ell)| \Big)^2\  \le \  \alpha \Big(  \E \big( D_\ell \big) \Big)^2
\end{equation*}
for $\ell$ sufficiently large (as a function of $\alpha$). This allows to rewrite \eqref{eq:var bound with L} in the simpler form stated in the proposition:
\begin{equation*}
  \Var\left( D_\ell \right)\le 241\cdot \alpha\cdot
  \Big(\E \big( D_\ell \big) \Big)^2.\qedhere
\end{equation*}
\end{proof}

\bigskip

\subsection{Putting it all together: $T=0$ for the nearest-neighbor case} \label{sec:nailing_T_0}\mbox{} \\[-1ex]

We now have all the tools for proving the assertion made in   Theorem~\ref{thm:power_law_bound} for zero temperature.

\begin{proof}[Proof of Theorem~\ref{thm:power_law_bound} at $T=0$]
Recall from~\eqref{eq:gamma def} that
\begin{equation}\label{eq:delta_choice}
  \gamma = 2^{-10}\chi\left(\frac{50 J}{\eps}\right)\, .
\end{equation}
In view of \eqref{eq:p_L_def}, if Theorem~\ref{thm:power_law_bound} does not hold at zero temperature for $J$ and $\eps$ then
\begin{equation*}
  \limsup_{L\to\infty} L^{\gamma}\cdot m(L)=\infty\, ,
\end{equation*}
which implies that
\begin{equation}\label{eq:contradiction_assumption}
  \text{$m(M)\ge M^{-\gamma}$ for infinitely many $M$}\, .
\end{equation}
We assume, in order to obtain a contradiction, that \eqref{eq:contradiction_assumption} holds. Let $M\ge 64$, later chosen sufficiently large, be such that $m(M)\ge M^{-\gamma}$. Applying Proposition~\ref{prop:comp_decay} we see that there is an $8\le \sqrt{M}\le L\le M$ such that
\begin{equation}\label{eq:controlled decay}
  m(L)\le m(j)\le m(L)\left(\frac{L}{j}\right)^{2\gamma} \,,\quad 1\le j\le L\, .
\end{equation}

We consider the three domains,  $\Lambda(k\ell)$ with $\ell = \lfloor L/4\rfloor$ and $k=1,2,3$.
Applying Proposition~\ref{prop:quantitative_lower_bound}
we obtain the anti-concentration inequality
\begin{equation*}
  \P\left(\frac{D_\ell}{\E \big( D_\ell \big)} < \frac{1}{2}\right)\,\ge\, \chi\left(\frac{4J}{\eps}\cdot\frac{|\extB\Lambda(2\ell)|}{\sqrt{|\Lambda(\ell)|}}\cdot \frac{m(\ell - 1)}{m(4\ell)} \right).
\end{equation*}

The right-hand side may be simplified, using \eqref{eq:controlled decay} together with the fact that $\gamma<2^{-10}$, and noting that the assumption $L\ge 8$ implies that $|\extB\Lambda(2\ell)|=4(2\ell+1)\le 3L$ and $|\Lambda(\ell)|=1+2\ell(\ell+1)\ge \frac{L^2}{16}$. This yields
\begin{equation}\label{eq:lower bound on fluctuations}
   \P\left(\frac{D_\ell}{\E \big( D_\ell \big)} < \frac{1}{2}\right) \,
   \ge\, \chi\left(\frac{4J}{\eps}\cdot \frac{3L}{L/4}\cdot\left(\frac{L}{\ell -1}\right)^{2\gamma}\right)\ge \chi\left(\frac{50J}{\eps}\right)\, .
\end{equation}

We shall now reach a contradiction by applying Proposition~\ref{prop:var_bound} with $\ell = \lfloor L/4\rfloor$, noting that the assumptions of that proposition are verified by \eqref{eq:controlled decay} and the fact that $\sqrt{M}\le L\le M$ and $m(M)\ge M^{-\alpha}$. The proposition implies that for $L$ sufficiently large (obtained by choosing $M$ sufficiently large), we have the concentration bound,
\begin{equation*}
 \Var\left( D_\ell \right)\le 241\cdot \gamma\cdot
  \Big(\E \big( D_\ell \big) \Big)^2\, .
\end{equation*}
Chebyshev's inequality then shows that
\begin{equation*}
    \P\left(\frac{D_\ell}{\E \big( D_\ell \big)} < \frac{1}{2}\right) \,
   \le \, 1000\gamma\, .
\end{equation*}
As this contradicts \eqref{eq:lower bound on fluctuations} for the choice \eqref{eq:delta_choice} of $\gamma$, we conclude that our initial assumption \eqref{eq:contradiction_assumption} must be false, implying that Theorem~\ref{thm:power_law_bound} holds at zero temperature.
\end{proof}

\bigskip

\section{Extension of the power-law upper bound to $T>0$}

\label{sec:positive_temperature}
In this section we adapt the zero-temperature proof of Theorem~\ref{thm:power_law_bound} to the positive temperature case. Again, for simplicity, we focus first on the case of nearest-neighbor interaction with the extension to finite-range interactions to follow in Section\,\ref{sec:rangeR}.

\subsection{Adjustments in the terminology} \mbox{} \\[-1ex]

At positive temperature the relevant function of the random field and of the boundary conditions is not the single ground-state configuration but the corresponding Gibbs probability measure. We proceed to explain how the proof is modified to account for this difference.

{\bf Influence/disagreement percolation.}
The order parameter, which at $T=0$ was the disagreement percolation  of \eqref{eq:p_L_def}
\begin{equation}
m(j; 0, \mathcal  J, h, \epsilon)  \=  \P\left(\sigma^{\Lambda(j),+}_{\mathbf{0}}> \sigma^{\Lambda(j),-}_{\mathbf{0}}\right)  \,
  \end{equation}
is replaced by the difference in the expected magnetization
\be
m(j; T, \mathcal  J, h, \epsilon)  \= \frac 1 2 \left[
    \E[\langle\sigma_\0\rangle^{\Lambda(j), +} ] \ - \      \E[\langle\sigma_\0\rangle^{\Lambda(j), -}]
    \right] \,.
\ee

Let us comment in passing that  the available monotone  coupling of the $+$ and $-$ probability measures allows to present also the last expression as the probability of disagreement percolation. However, to keep the discussion simple, we shall not stress this point.  \\

Correspondingly, as a measure of the disagreement in $\Lambda(\ell)$ due to the difference in boundary conditions placed on $\Lambda(3\ell)$ we take
\be \label{eq:D_l_def_pos_temp}
D_\ell(\eta) \ := \   \frac{1}{2}\sum_{v\in \Lambda(\ell)} \left[\langle\sigma_v\rangle^{\Lambda(3\ell), +}-\langle\sigma_v\rangle^{\Lambda(3\ell), -}\right]\, .
\ee
Recall also that, at $T=0$, $B_\ell(\eta)/J$ counted the number of edges in the separating surface $\partial_{\text{e}}\Lambda(2\ell)$ which contribute to the surface tension. At $T>0$, we find it more convenient to count vertices rather than edges, leading to the definition
\be \label{eq:B_l_def_pos_temp}
  \tilde{B}_{\ell}(\eta)   \ := \frac{J}{2}\sum_{v\in\extB\Lambda(2\ell)} \left[\langle\sigma_v\rangle^{\Lambda(3\ell)\setminus\Lambda(\ell),+} - \langle\sigma_v\rangle^{\Lambda(3\ell)\setminus\Lambda(\ell),-}\right].
\ee

{\bf Surface tension.}
For $T>0$, the role which is played by energy in the zero-temperature analysis is taken by the free energy, which for different combinations of the boundary conditions is defined as:
\begin{equation}
  \mathcal{F}^{s,s'}_\ell  := -T\cdot\log(Z^{\Lambda(3\ell)\backslash \Lambda(\ell); s,s'})
\end{equation}
where  $s$ and $s'$ indicate the $(\pm)$ boundary conditions placed on the external boundary  of $\Lambda(3\ell)$ and the internal boundary of $\Lambda(\ell)$, respectively, and the partition function is
 \begin{equation}
    Z^{\Lambda(3\ell)\backslash \Lambda(\ell); s,s'}=Z^{s,s'}=\sum_{\sigma:\Lambda(3\ell)\backslash \Lambda(\ell)\to\{-1,1\}} \exp\left(-\frac 1T H^{\Lambda(3\ell)\backslash \Lambda(\ell); s,s'}(\sigma)\right)
  \end{equation}
with $H^{\Lambda(3\ell)\backslash \Lambda(\ell); s,s'} = H^{s,s'}$ the Hamiltonian incorporating the boundary conditions.

Following this prescription,  the extension of  the surface tension, of \eqref{T_def},  to positive temperatures is
\be
  \mathcal{T}_{\ell}(\eta)  \=   T  \, \log\left(
    \frac{Z^{+,+} \cdot Z^{-,-} } {Z^{+,-} \cdot Z^{-,+} }\right)\,.
\ee
A similar replacement takes place in the definition of the function $G_{\ell}(\eta)$ in~\eqref{eq:G_def} and it is straightforward to check that the relation~\eqref{eq:TandG} still holds.\\

\subsection{Extension of the proof to $T>0$} \mbox{} \\[-1ex]

The zero-temperature bound of  Theorem~\ref{thm:T1}
is modified into the following statement, in which we replace the references to the ground-state spins by their quenched averages and where, for simplicity,
we have upper bounded a sum over $(u,v)\in\partial_{\text{e}}\Lambda(2\ell)$ (analogous to the one in Theorem~\ref{thm:T1})
by a sum over $v\in\extB\Lambda(2\ell)$.

\begin{thm}\label{thm:upper bound surface tension_pos_temp}
In the RFIM with nearest-neighbor interaction, for any realization of the field $\eta$,
\begin{equation} \label{Tpos_surface_tension}
\mathcal{T}_{\ell}(\eta)\le 8 \tilde{B}_\ell(\eta)\,.
\end{equation}
\end{thm}
\begin{proof}
As in the $T=0$ case, the set  $\extB\Lambda(2\ell)$ enters the discussion as a  separating barrier between
the inner and the outer boundary of  $\Lambda(3\ell) \backslash \Lambda(\ell)$.  Denoting   the restriction of the spin configuration to this set by  $\tau:\extB\Lambda(2\ell)\to\{-1,1\}$, let $\rho_+$ and, correspondingly, $\rho_-$ be the two probability measures induced on it by the $(+,+)$ and $(-,-)$ boundary conditions.  More explicitly,
  \be
  \rho_+(\tau) \= \frac{Z^{+,+}_\tau}{Z^{+,+}}\,, \, \qquad   \rho_-(\tau) \= \frac{Z^{-,-}_\tau}{Z^{-,-}}\, ,
  \ee
with $Z^{s,s'}_\tau$ the restricted partition functions
  \begin{equation*}
    Z^{s,s'}_\tau:=\sum_{\substack{\sigma:\Lambda(3\ell)\backslash \Lambda(\ell)\to\{-1,1\}\\
    \sigma|_{\extB\Lambda(2\ell)} = \tau}} \exp\left(-\frac 1T H^{s,s'}(\sigma)\right).
  \end{equation*}

Considering first the $(+)$ case, let us note that
  \be
 \sum_{\tau:\extB\Lambda(2\ell)\to\{-1,1\}} \rho_+(\tau) \,  \frac{Z^{+,-}_\tau}{Z^{+,+}_\tau} \=  \frac{Z^{+,-}}{Z^{+,+}}
  \ee
 Hence, by Jensen's inequality (and the convexity of $-\log(X)$),
 for each specified $\eta$ (which is omitted in the following expression)
 \be
   \log\left(\frac{Z^{+,+}}{Z^{+,-}}\right) \
 \leq  \  \sum_{\tau:\extB\Lambda(2\ell)\to\{-1,1\}}
   \, \rho_+(\tau)  \, \log\left(\frac{Z^{+,+}_\tau}{Z^{+,-}_\tau}\right) \
   \ee

  Combining the above with the analogous statement for $\rho_-(\tau) $ we get:
     \begin{eqnarray}\label{eq:tau_eta_first_upper_bound}
    \mathcal{T}_{\ell}(\eta) & = &   T  \log\left(
    \frac{Z^{+,+}}{Z^{+,-}}\cdot \frac{Z^{-,-}}{Z^{-,+}}\right)  \ \leq  \  \\
    &  \leq &  T\left[\sum_{\tau:\extB\Lambda(2\ell)\to\{-1,1\}}
   \, \rho_+(\tau)
    \, \log\left(\frac{Z^{+,+}_\tau}{Z^{+,-}_\tau}\right) \  +  \
    \sum_{\tau:\extB\Lambda(2\ell)\to\{-1,1\}}
   \, \rho_-(\tau) \log\left(\frac{Z^{-,-}_\tau}{Z^{-,+}_\tau}\right)\right]\,.
   \notag
  \end{eqnarray}

  We now use the fact that the measure $\P^{+,+}$ stochastically dominates $\P^{-,-}$, as in \eqref{eq:stochastic_domination_pos_temp}. In particular, there exists a probability measure $\rho(\tau^+,\tau^-)$ on pairs $\tau^+,\tau^-:\extB\Lambda(2\ell)\to\{-1,1\}$ such that $\tau^+\ge\tau^-$ pointwise, with probability $1$, and the marginal distribution of each $\tau^s$ is given by $\rho_s$.  This \emph{coupling}  of measures allows to express  \eqref{eq:tau_eta_first_upper_bound} in the form
  \begin{equation}\label{eq:tau_eta_second_upper_bound}
    \mathcal{T}_\ell(\eta)\le T\left[\sum_{\tau^+,\tau^-:\extB\Lambda(2\ell)\to\{-1,1\}}
    \rho(\tau^+,\tau^-)
    \log\left(\frac{Z^{+,+}_{\tau^+}}{Z^{+,-}_{\tau^+}}\cdot \frac{Z^{-,-}_{\tau^-}}{Z^{-,+}_{\tau^-}}\right)\right].
  \end{equation}

The coupling of the measures allows to bound the quantity on the right in terms of the positive temperature version of the disagreement percolation.
The estimate is motivated by the observation that for every configuration $\tau$:
  \be \label{cross_ratio_at_tau}
Z^{+,+}_{\tau}\cdot Z^{-,-}_{\tau}\= Z^{+,-}_{\tau} \cdot Z^{-,+}_{\tau}\,.
   \ee
The proof is through the  bijection associating  to each pair $(\sigma^{+,+},\sigma^{-,-})$ contributing to the double sum on the left   the following pair $(\sigma^{+,-},\sigma^{-,+})$ contributing to the double sum on the right:
  \begin{equation}\label{eq:bijection}
    {\sigma}^{+,-}_v:=\begin{cases}
     \sigma^{+,+}_v&v\in\Lambda(3\ell)\setminus\Lambda(2\ell)\\
     \sigma^{-,-}_v&v\in\Lambda(2\ell)\setminus\Lambda(\ell)
    \end{cases},\quad
   {\sigma}^{-,+}_v:=\begin{cases}
     \sigma^{-,-}_v&v\in\Lambda(3\ell)\setminus\Lambda(2\ell)\\
     \sigma^{+,+}_v&v\in\Lambda(2\ell)\setminus\Lambda(\ell)
    \end{cases}.
  \end{equation}
At the common value of the configuration $\tau$ over the separating set $\extB \Lambda(2\ell)$,
    the  sums of the corresponding energy terms in \eqref{cross_ratio_at_tau} match.

    Thus terms with $\tau^+=\tau^-$ make no contribution to the sum
    \eqref{eq:tau_eta_second_upper_bound}.  For the  more general case
we note that when the restriction of $\sigma^{+,+}$ ($\sigma^{-,-}$) to $\extB\Lambda(2\ell)$ is $\tau^+$ ($\tau^-$) and $\sigma^{+,-}, \sigma^{-,+}$ are given by \eqref{eq:bijection} then, with $\tau^+\ge\tau^-$,
  \begin{equation}\label{eq:energy_difference_bound}
  \begin{split}
    &-\frac{1}{T}\left(H^{+,+}(\sigma^{+,+}) + H^{-,-}(\sigma^{-,-}) - H^{+,-}(\sigma^{+,-}) - H^{-,+}(\sigma^{-,+})\right)\\
    &=\frac{J}{T}\sum_{(u,v)\in\partial_{\text{e}}\Lambda(2\ell)} \left(\sigma^{+,+}_u\sigma^{+,+}_v + \sigma^{-,-}_u\sigma^{-,-}_v - \sigma^{+,+}_u\sigma^{-,-}_v - \sigma^{-,-}_u\sigma^{+,+}_v\right)\\
    &=\frac{J}{T}\sum_{(u,v)\in\partial_{\text{e}}\Lambda(2\ell)} \left(\sigma^{+,+}_u - \sigma^{-,-}_u\right)\cdot\left(\sigma^{+,+}_v - \sigma^{-,-}_v\right)\\
    &=\frac{J}{T}\sum_{(u,v)\in\partial_{\text{e}}\Lambda(2\ell)} \left(\sigma^{+,+}_u - \sigma^{-,-}_u\right)\cdot\left(\tau^+_v - \tau^-_v\right)\le \frac{4J}{T}\sum_{v\in\extB\Lambda(2\ell)} \left(\tau^+_v - \tau^-_v\right),
  \end{split}
  \end{equation}
where the third equality uses the fact that if $(u,v)\in\partial_{\text{e}}\Lambda(2\ell)$ then $v\in\extB\Lambda(2\ell)$ and the inequality uses the fact that each vertex $v\in\extB\Lambda(2\ell)$ is incident to at most two edges $(u,v)\in\partial_{\text{e}}\Lambda(2\ell)$ and the fact that $\tau^+\ge \tau^-$ pointwise.
Thus   \begin{equation*}
    \frac{Z^{+,+}_{\tau^+}}{Z^{+,-}_{\tau^+}}\cdot \frac{Z^{-,-}_{\tau^-}}{Z^{-,+}_{\tau^-}}\le \exp\left(\frac{4J}{T}\sum_{v\in\extB\Lambda(2\ell)} \left(\tau^+_v - \tau^-_v\right)\right).
  \end{equation*}
  Finally, inserting this estimate in
   \eqref{eq:tau_eta_second_upper_bound} we get
  \begin{equation*}
    \mathcal{T}_\ell(\eta)\le 4J\sum_{\tau^+,\tau^-:\extB\Lambda(2\ell)\to\{-1,1\}}\rho(\tau^+,\tau^-)\sum_{v\in\extB\Lambda(2\ell)} \left(\tau^+_v - \tau^-_v\right).
  \end{equation*}
Through the definition of $\rho(\tau^+,\tau^-)$ the above reduces to the bound asserted in \eqref{Tpos_surface_tension}.
\end{proof}

The representation of the surface tension given by Theorem~\ref{thm:T2}, which enables a lower bound on its expected value at zero temperature, continues to hold at positive temperature with the exact same statement. The proof also remains the same, upon replacing~\eqref{eq:G_derivative_formula} and~\eqref{eq:surface_tension_via_derivative} with the analogous
\begin{equation}
  \frac{\partial}{\partial\eta_v} G_\ell(\eta) = \eps
  \left[ \langle\sigma_v\rangle^{\Lambda(3\ell), +}-\langle\sigma_v\rangle^{\Lambda(3\ell), -} \right]
\end{equation}
and
\begin{equation}
\mathcal T_\ell(\eta) = \eps\int_{-\infty}^\infty
\sum_{v \in \Lambda(\ell)} \left[ \langle\sigma_v\rangle^{\Lambda(3\ell), +}(\eta^{(t)})-\langle\sigma_v\rangle^{\Lambda(3\ell), -}(\eta^{(t)}) \right] \, dt\, = \, 2\eps
\int_{-\infty}^\infty D_{\ell}(\eta^{(t)})\, dt\, .
\end{equation}

Combining Theorem~\ref{thm:upper bound surface tension_pos_temp} and~\eqref{DD2} we obtain
\begin{eqnarray}
 \frac{4 \,  \E(\tilde{B}_\ell(\eta))}{\eps \sqrt{|\Lambda(\ell)|}}
 & \geq &   \E \left(  \frac{D_{\ell}(\eta) }{| \Lambda(\ell)|}\,
 \frac{1 }{ \phi(\widehat \eta) }  \right) \,
\end{eqnarray}
which replaces \eqref{101} when $T>0$. The bound implies that Proposition~\ref{prop:quantitative_lower_bound} continues to hold at positive temperature, with the exact same statement and with $2\tilde{B}_\ell$ replacing $B_\ell$ throughout the proof (noting, in particular, that
\begin{equation}
  2\E(\tilde{B}_\ell) \le 2\, J\, |\extB\Lambda(2\ell)|\,m(\ell-1)
\end{equation}
holds instead of \eqref{eq:B_ell_bound}).

The upper bound on the variance of $D_\ell$, given for $T=0$ by Proposition~\ref{prop:var_bound}, continues to hold exactly as stated also when $T>0$. In the proof, the indicator random variable of the event $E_v$ is replaced with the random variable
\begin{equation}
  X_v := \frac{1}{2}\left[\langle\sigma_v\rangle^{\Lambda(3\ell), +}-\langle\sigma_v\rangle^{\Lambda(3\ell), -}\right]\,.
\end{equation}
This yields, e.g., the analogous equation to~\eqref{eq:var_as_sum_over_pairs},
\begin{equation}
  \Var\left(D_\ell \right) = \sum_{v,w\in \Lambda(\ell)}  \big[\E(X_v\cdot X_w) - \E(X_v)\E(X_w)\big]\,
\end{equation}
and the analogous equation to~\eqref{eq:uv},
\be
  \E(X_v\cdot X_w) \< m(r(v,w))^2
\end{equation}
with $r(v,w)$ defined in \eqref{eq:r_v_w_def}. The last inequality holds as, via the monotonicity property~\eqref{eq:FKG_consequence_pos_temp},
\begin{equation*}
  \E(X_v\cdot X_w) \le \frac{1}{4}\E\left[\left(\langle\sigma_v\rangle^{\Lambda_v(r(v,w)), +}-\langle\sigma_v\rangle^{\Lambda_v(r(v,w)), -}\right)\left(\langle\sigma_v\rangle^{\Lambda_w(r(v,w)), +}-\langle\sigma_v\rangle^{\Lambda_w(r(v,w)), -}\right)\right]\,
\end{equation*}
after which one may rely on independence.

The end of the proof of Theorem~\ref{thm:power_law_bound}, detailed in Section~\ref{sec:nailing_T_0} for the zero-temperature case, applies without change to prove the theorem at positive temperature.

\bigskip

\section{Extension to   finite-range interactions}
\label{sec:rangeR}

At $T=0$, the proof for general finite-range interactions $\mathcal J$ remains the same with the following minor changes, in which $C_k(\mathcal{J})$ denote positive constants depending only on $\mathcal J$ and $R(\mathcal J) = \max\{d(u,v) \, : \, J_{u,v} \neq 0\}$ (the interaction's range).
\begin{enumerate}[1)]
  \item The statement of Theorem~\ref{thm:T1} is changed by replacing the bound $B_{\ell} (\eta)  \leq   2 J \, |\extB \Lambda(2 \ell)|$ by
  \begin{equation}
    B_{\ell} (\eta) \  \leq \  \sum_{(u,v)\in \partial_{\text{e}} \Lambda(2\ell)} J_{u,v}\,.
  \end{equation}
  \item The condition $|h+\eps \eta_v| >  4J$ appearing in the proof of Theorem~\ref{thm:T2} is replaced by $|h+\eps \eta_v| >  \sum_v J_{\zero, v}$.
  \item\label{item:prop_change} The bound~\eqref{eq:B_ell_bound} is replaced by
  \begin{equation*}
    \E(B_\ell) \le C_1(\mathcal{J})|\extB\Lambda(2\ell)|\,m(\ell-R(\mathcal J)).
  \end{equation*}
  Consequently in~\eqref{202}, $4 J |\extB\Lambda(2\ell)|\,m(\ell-1)$ becomes $C_2(\mathcal{J}) |\extB\Lambda(2\ell)|\,m(\ell-R(\mathcal J))$.
  \item In the proof of Proposition~\ref{prop:var_bound}, the definition of $r(v,w)$ in \eqref{eq:r_v_w_def} is replaced by
  \begin{equation*}
    r(v,w):=\lfloor (d(v,w)-R(\mathcal J))/2\rfloor.
  \end{equation*}
  The simple bound~\eqref{eq:E_v_E_w_simple_bound} is then used for pairs $v,w$ at distance $d(v,w)\le R(\mathcal J) + 1$, leading to the factor $|\Lambda(2)|$ appearing in the proof being replaced by $|\Lambda(R(\mathcal J) + 1)|$.

  The statement of Proposition~\ref{prop:var_bound} is changed to allow $L_0$ to depend on $\mathcal J$ (besides $\alpha$).
  \item The proof of Theorem~\ref{thm:power_law_bound} given in Section~\ref{sec:nailing_T_0} is modified by taking into account the change described in item~\ref{item:prop_change} above in the constants appearing in Proposition~\ref{prop:quantitative_lower_bound}. Correspondingly, inequality~\eqref{eq:lower bound on fluctuations} is modified to
      \begin{equation}
        \P\left(\frac{D_\ell}{\E \big( D_\ell \big)} < \frac{1}{2}\right) \,
        \ge\, \chi\left(\frac{C_3(\mathcal J)}{\eps}\cdot \left(\frac{L}{\ell - R(\mathcal J)}\right)^{2\gamma}\right)\ge \chi\left(\frac{C_4(\mathcal J)}{\eps}\right)
      \end{equation}
      holding for $L$ sufficiently large, and the power $\gamma$ appearing in the theorem is modified from its value in~\eqref{eq:delta_choice} to
      \begin{equation}
        \gamma = 2^{-10}\chi\left(\frac{C_4(\mathcal J)}{\eps}\right).
      \end{equation}

  \medskip
  At $T>0$, the argument extends to general finite-range interactions by applying the following changes:
  \begin{enumerate}[1)]
  \item The definition of $\tilde{B}_{\ell}$ in \eqref{eq:B_l_def_pos_temp} is modified to
  \be
  \tilde{B}_{\ell}(\eta)\ := \frac{1}{4}\sum_{v\in\extB\Lambda(2\ell)} J_v \left[\langle\sigma_v\rangle^{\Lambda(3\ell)\setminus\Lambda(\ell),+} - \langle\sigma_v\rangle^{\Lambda(3\ell)\setminus\Lambda(\ell),-}\right]
   \ee
   with
   \be
    J_v\ :=\ \sum_{u\colon (u,v)\in\partial_{\text{e}}\Lambda(2\ell)} J_{u,v}\,.
   \ee

  \item The proof of Theorem~\ref{thm:upper bound surface tension_pos_temp} is modified by replacing the inequality~\eqref{eq:energy_difference_bound} with
  \begin{equation}
  \begin{split}
    &-\frac{1}{T}\left(H^{+,+}(\sigma^{+,+}) + H^{-,-}(\sigma^{-,-}) - H^{+,-}(\sigma^{+,-}) - H^{-,+}(\sigma^{-,+})\right)\\
    &=\frac{1}{T}\sum_{(u,v)\in\partial_{\text{e}}\Lambda(2\ell)} J_{u,v}\left(\sigma^{+,+}_u\sigma^{+,+}_v + \sigma^{-,-}_u\sigma^{-,-}_v - \sigma^{+,+}_u\sigma^{-,-}_v - \sigma^{-,-}_u\sigma^{+,+}_v\right)\\
    &=\frac{1}{T}\sum_{(u,v)\in\partial_{\text{e}}\Lambda(2\ell)} J_{u,v}\left(\sigma^{+,+}_u - \sigma^{-,-}_u\right)\cdot\left(\sigma^{+,+}_v - \sigma^{-,-}_v\right)\\
    &=\frac{1}{T}\sum_{(u,v)\in\partial_{\text{e}}\Lambda(2\ell)} J_{u,v}\left(\sigma^{+,+}_u - \sigma^{-,-}_u\right)\cdot\left(\tau^+_v - \tau^-_v\right)\le \frac{2}{T}\sum_{v\in\extB\Lambda(2\ell)} J_v\left(\tau^+_v - \tau^-_v\right),
  \end{split}
  \end{equation}
  with this change propagating to the next two displayed equations in the proof.

  \item The changes analogous to those described for the $T=0$ case.
  \end{enumerate}
\end{enumerate}

\bigskip

 \section{Magnetization decoupling bounds}   \label{app:corr}
For completeness sake we enclose here proofs
that
the influence percolation probability $m(\ell,...)$ provides  bounds on both
the covariance between the quenched  local magnetizations at distant sites
and
the spin - spin covariance within the Gibbs states at typical configurations of the random field, as was asserted in  \eqref{Es1}  and \eqref{eq:mcorr}. The arguments apply in the generality of the random-field Ising model on a general infinite transitive graph, in any of its infinite-volume Gibbs states.

\begin{lemma}\label{lem:cov} In the random field Ising model on a transitive graph, with spin-spin coupling of a finite range $R(\mathcal J)$ and any pair of vertices $\{u,v\}$ at distance $d(u,v)$. If $ d(u,v) > \ell$ then
\be \label{Es1_app}
 \E{(\langle \sigma_u;\sigma_v\rangle )} \  \leq   \   2\,  m(\ell ;T,  \mathcal J, h, \epsilon),
\ee
while if $ d(u,v) \geq 2\ell+R(\mathcal J)$ then
\be
  \label{eq:mcorr_app}
\rm{Cov}\Big( \langle \sigma_u \rangle; \langle \sigma_v\rangle) \Big) \ := \
   \E(\langle \sigma_u \rangle; \langle \sigma_v\rangle) \ \leq \  4 \,   m(\ell ;T,  \mathcal J, h, \epsilon) \,.
\ee
\end{lemma}

\noindent{\bf Proof}:
{\it i)\/}  By the FKG monotonicity of the RFIM Gibbs states, the Gibbs conditional expectation of
$\sigma_u$, conditioned on the configuration's restriction to the complement of the set $\Lambda_u(\ell)$, satisfies, for any configuration of the random field
\be
\langle\sigma_u\rangle^{\Lambda_u(\ell), -}  \ \leq \
 \langle\sigma_u\rangle^{\Lambda_u(\ell), \sigma_{\Lambda_u(\ell)^c}}  \ \leq
 \
 \langle\sigma_u\rangle^{\Lambda_u(\ell), +} \,.
\ee
Averaging over $\sigma_{\Lambda_u(\ell)^c} $, one learns that also the infinite-volume expectation value is bracketed by
$\langle\sigma_u\rangle^{\Lambda_u(\ell), \pm} $:
\be \label{diff4}
\langle\sigma_u\rangle^{\Lambda_u(\ell), -}  \ \leq \
 \langle\sigma_u\rangle  \ \leq
 \
 \langle\sigma_u\rangle^{\Lambda_u(\ell), +} \,.
\ee
The two equations imply:
\be \label{diff5}
\Big |  \langle\sigma_u\rangle  \ - \   \langle\sigma_u\rangle^{\Lambda_u(\ell), \sigma_{\Lambda_u(\ell)^c}}  \Big | \ \leq \
 \Big[   \langle\sigma_u\rangle^{\Lambda_u(\ell), +} -
 \langle\sigma_u\rangle^{\Lambda_u(\ell), -}  \Big ]
\ee
The covariance of the spins within the infinite-volume Gibbs state can be written as
\ba \label{cov_expression}
 \langle \sigma_u;\sigma_v\rangle  & =
 \langle \big(  \sigma_u -   \langle\sigma_u\rangle \big) \, \, \sigma_v\rangle
 \\
 & =
  \langle \big( \langle\sigma_u\rangle^{\Lambda_u(\ell), \sigma_{\Lambda_u(\ell)^c}}   -   \langle\sigma_u\rangle\big) \, \, \sigma_v\rangle
\ea
where the second equation is by the state's Dobrushin-Lanford-Ruelle property and the assumption that $ d(u,v) > \ell$.

Combining \eqref{cov_expression}  with \eqref{diff5} we learn that for any realization of the random field
\be \Big| \langle \sigma_u;\sigma_v\rangle  \Big| \  \leq \
\Big[   \langle\sigma_u\rangle^{\Lambda_u(\ell), +} -
 \langle\sigma_u\rangle^{\Lambda_u(\ell), -}  \Big ].
\ee
Averaging this relation over the disorder one gets \eqref{Es1_app}. \\

ii)  For the second covariance bound let
\be
 \langle\sigma_u\rangle^{\Lambda_u(\ell), av}  \ := \
\frac 1 2\Big[   \langle\sigma_u\rangle^{\Lambda_u(\ell), +} +
 \langle\sigma_u\rangle^{\Lambda_u(\ell), -}  \Big]
\ee
and observe that since the random fields on which $\langle\sigma_u\rangle^{\Lambda_u(\ell), av}$ and $\langle\sigma_v\rangle^{\Lambda_v(\ell), av}$  depend belong to disjoint sets, their covariance vanishes:
\be
 \rm{Cov}\Big(\langle\sigma_u\rangle^{\Lambda_u(\ell), av},\langle\sigma_u\rangle^{\Lambda_u(\ell), av}\Big) \= 0
\ee
Furthermore, by \eqref{diff4},
\be  \label{diff6}
\Big |  \langle\sigma_u\rangle  \ - \   \langle\sigma_u\rangle^{\Lambda_u(\ell), av}   \Big | \ \leq \
\frac 1 2\Big[   \langle\sigma_u\rangle^{\Lambda_u(\ell), +} -
 \langle\sigma_u\rangle^{\Lambda_u(\ell), -}  \Big ] \,. \\[2ex]
\ee

The claimed  \eqref{eq:mcorr_app} then follows by a simple application of the general covariance bound:
\begin{multline}
\Big | \rm{Cov}(A,B) - \rm{Cov}(\widetilde A, \widetilde B)  \Big  | = \Big |\E[(A - \widetilde A)B]+\E[\widetilde A(B-\widetilde B)]+\E[(\widetilde A - A)]\E [B]+\E[\widetilde A]\E[(\widetilde B - B)]\Big  | \\
\leq 2 \,\|A-\widetilde A\|_1 \cdot \| B \|_{\infty} \ + \ 2 \, \|B-\widetilde B\|_1 \cdot \| \widetilde A \|_{\infty}
\end{multline}
applied to
\ba A &= \langle\sigma_u\rangle\,  ,  \qquad \widetilde A&=\langle\sigma_u\rangle^{\Lambda_u(\ell), av} \\
B&= \langle\sigma_v\rangle\,  ,  \qquad \widetilde B &=\langle\sigma_v\rangle^{\Lambda_v(\ell), av}
\ea
for which, by \eqref{diff6} and the definition of $m(\ell ;T,  \mathcal J, h, \epsilon) $,
\be
\|A-\widetilde A\|_1 = \|B-\widetilde B\|_1 \ \leq \   m(\ell ;T,  \mathcal J, h, \epsilon).
\ee
\qed

\medskip

\section{Discussion and open questions}  \label{sec:OP}

In summary: our study quantifies the analysis of~\cite{AW89,AW90} that for each value of the external field  the model's Hamiltonian almost surely has a unique infinite-volume ground state, and similarly unique positive-temperature Gibbs states.  The upper bounds proven here
establish that  the probability that the ground-state configuration depends on the quenched disorder at distance $\ell$ away  decays by at least an $\eps$-dependent  power, and exponentially fast if the disorder parameter is sufficiently large.  However, our understanding of the model remains incomplete.  Following is a selection of open questions, some with relevance for physics models and some as a challenge to probabilists of related interests.  \\

{\bf Exponential vs. power-law decay.}  As mentioned above, an open question of enduring interest is whether as the disorder parameter ($\eps/ J$) is tuned down the ground state's dependence on the quenched disorder  makes a transition from exponential decay to a power law. Tentative but admittedly weak arguments have appeared for each of these possibilities (\cite{GM82, BK88} and \cite{DS84}). Also of interest is the corresponding question for the $O(N)$ symmetric models in dimensions $d\leq 4$, the latter being the critical dimension for the Imry-Ma phenomenon in the presence of continuous symmetry.   \\

{\bf  Cluster dynamics.}  Consider the RFIM dynamics in which a large system with a quenched random field is subject to a slowly varying uniform magnetic field $h$.  For $|h|/\eps$ large enough, the ground state configuration is close to being constant, coinciding with   the sign of $h$.
As the uniform field is increased, starting from the sufficiently  negative value,  the corresponding ground state configuration changes in a sequence of flips, in which a cluster of $-$ spins  flips to $+$ spins.  Thus the graph is partitioned into connected clusters of sites for which at the given random field $\eta$ the spins flip at a common value of $h$.  It can be shown that in two dimensions almost surely each  flip involves only a finite number of sites, and the mean  value of the size of the cluster which flips along with a preselected site is finite throughout the regime in which the ground state spins decorrelate   exponentially fast.   Does the mean stay finite  for arbitrarily small $\eps>0$? \\


{\bf RFIM with other random field distributions.}  Our analysis  focused  on IID Gaussian disorder.
In contrast, the theorem of~\cite{AW89,AW90} applies to a wide class of random field distributions.
The Gaussian structure allowed a short-cut in the proof of Theorem~\ref{thm:T2}.
While we expect the results to be valid also well beyond this case,  that is not done here.

Among the other distributions of interest are:
\begin{enumerate}
\item  A dilute coercive field, with $(\eta_v)$ given by independent random variables with $\P(\eta_v=-\infty)=\P(\eta_v=\infty)=\eps$ and $\P(\eta_v = 0) = 1-2\eps$.   \\
This distribution was considered in \cite{DS84} where an observation was initially made  suggesting the possibility of a transition from exponential to  power-law decay of correlations at low $\eps$.  (However, subsequent considerations have weakened the case for that, cf. also the discussion in \cite{BK88}).\\

\item Bounded variables, e.g. with  $(\eta_v)$ independent and uniformly distributed   in $\{-1,1\}$ or $[-1,1]$.
The former is of particular relevance for  the case of $Q$-state Potts models with random couplings, for which $\sigma_v$ takes values in $\{1,..., Q\}$ and the Hamiltonian is:
\be
H_\eta (\sigma) \  =\  - \sum_{\{x,y\} \in E(\Z^2)} \left(J + \eps \eta_{x,y}\right)
\1[\sigma_x = \sigma_y ]  \,  .
\ee
The uniform bound on $|\eta|$ allows to keep the discussion separate from that of frustration effects. \\
\end{enumerate}

{\bf   The more general Imry-Ma phenomenon.}
While the RFIM is a bellwether for the more general Imry-Ma phenomenon,  the general case
is a bit more complicated on two accounts.  The first is the lack of a-priori obvious pair of opposing boundary conditions for the definition of the order parameter.   That can be addressed, as was done in \cite{AW90},  by inducing the $\pm$  states not through boundary conditions but throughout a mild shift of the uniform field beyond the corresponding boundary of the region under study, $h\to h\pm \delta h$ with $\delta h$ in the range
\be
  |\Lambda(\ell)|^{-1}\ll \delta h \ll 1 \,  \quad \mbox{ (as $ \ell \to \infty $)}  \,.
  \ee
(An alternative is to define the order parameter though a maximization of the difference induced by different boundary spin configurations.) A potentially more substantial difference with the RFIM, is that in the general case the natural order parameter does not control the difference in the configurations, or measures, just in their (generalized) magnetizations.  The resolution of this  complication may require some new technical  ideas.  \\

 \appendix

\section{Exponential decay at high  disorder}\label{sec:high disorder}

As a rule of thumb it is generally expected that at high enough disorder, be it thermal or due to noisy environment, correlations decay exponentially fast.  Results in this vein for systems related to the RFIM can be found in the works of A. Berretti~\cite{Ber85}, J. Imbrie and J. Fr\"ohlich~\cite{ImbFro85}, and F.  Camia, J.  Jiang and C.M. Newman~\cite{CJN18}.

Let us present here an especially simple proof of such behavior for the $T=0$ case, i.e. exponential decay of the correlations of the RFIM's ground state, and also of the principle that fast enough power-law decay implies exponential decay.

\begin{thm}  For the RFIM on $\Z^d$ with the nearest-neighbor interaction~\eqref{J_nn} and random field given by IID random variables $(\eta_u)$,  if
\be \label{large_eps}
\Pr\big(|h+\eps \eta_\0| \leq 2d J \big) \ < \ p_c(d) \,
\ee
with $p_c(d)$ the critical density for site percolation on $\Z^d$,
then $m(L;0,\mathcal J, h,\eps)$ decays exponentially fast in $L$.
\end{thm}
\begin{proof}
At sites where $|h+\eps \eta_v| > 2d J$  the ground-state configuration is dictated by the sign of the local field.  Hence disagreement percolation can propagate only along the sites with $|h+\eps\eta_v| \le  2d J$.
In the regime described by \eqref{large_eps} the exceptional sites form a sub-percolating point process, for which the connectivity probability is known to decay exponentially in the distance~\cite{AB87,Men86}.
\end{proof}

A boosted version of the above  simple argument allows to conclude that if on some scale $\ell $ the probability of influence propagation is small enough ($1/\ell^{d-1}$ power law with a small prefactor) then on larger scales the influence decays exponentially fast.
An analogous statement holds also for $T>0$, but for simplicity of presentation we present the proof for $T=0$.

\begin{thm} \label{thm:exp1}
For the RFIM on $\Z^d$ with the nearest-neighbor interaction~\eqref{J_nn}, there is a finite constant $c_0$ (depending only on $d$) with which: if for some $\ell < \infty$
\be
m(\ell;0,\mathcal J, h,\eps) \, \leq \,  c_0  / \ell^{d-1}\label{eq:one_over_ell_bound}
\ee
then for all $L <\infty$
\be
m(L;0,\mathcal J, h,\eps) \, \leq \,  C_1 \,  e^{- b  L/ \ell}
\ee
with  $C_1, b \in (0,\infty)$ which do not depend on $J$, $h$, $\epsilon$ and $\ell$.
\end{thm}

In particular, we learn that   $m(L;0,\mathcal J, h,\eps)$ cannot decay by a power law faster than $1/L$ without decaying  exponentially.

\begin{proof}
In the following we say that a site $v\in \Z^2$ is sensitive to
boundary conditions at distance $\ell$ if
\be \label{event}
\sigma^{\Lambda_v(\ell),+}_v  \neq \sigma^{\Lambda_v(\ell),-}_v  \,. \ee

For each $L> \ell$
the event whose probability defines $  m(L) $,
\be
\langle\sigma_\0\rangle^{\Lambda(L), +}  \ \neq  \       \langle\sigma_\0\rangle^{\Lambda(L), -}  \, ,
\ee
requires the existence of a path  from $\0$  to the set $\extB \Lambda(L-\ell)$   along
sites $v\in \Z^2$  at which the  condition \eqref{event} holds.

Let now $\mathcal {P}_\ell$ be a partition of the vertex set of $\Z^d$ into a $\Z^d$-like array of disjoint cubic blocks of side length  $2\ell$, and consider the random set of blocks in this partition which  contain at least one site for which  \eqref{event} holds.
These block events are  1-step independent, in the sense that they are jointly independent for any collection of blocks of which no two are touching.

The probability that an individual  block contains a site at which the  condition \eqref{event} holds is trivially dominated by $|\extB \Lambda(\ell)| \times m(\ell)$.  Adjusting the  constant $c_0$ in assumption~\eqref{eq:one_over_ell_bound} the above probability
can be made as small as convenient.
The claim then follows through a standard exponentially-decaying bound on the connectivity probability in 1-step independent percolation of small enough density.
 \end{proof}

\medskip

\section{The Mandelbrot percolation analogy} \label{app:man} \mbox{}\\[-1ex]
The results presented above do not answer the question whether in two dimensions the exponential decay of correlations persists into arbitrarily small values of the disorder parameter, or whether the exponential decay turns into a power-law decay at low enough (but still non zero) $\eps$. Related to this is the question of what  would be a sensible algorithm for the computation of the  ground state $\widehat \sigma$ for a given random field,  and how would it perform  at very low disorder.

An intriguing perspective is provided by the following hierarchal algorithm. It has the virtue of simplicity but also the drawback of being potentially misleading through over simplification. It is formulated for the specific case $h=0$ and nearest-neighbor interaction.

Let $(\mathcal P_n)$, $n\ge 0$, be a sequence of nested partitions of $\Z^2$ into square blocks, with the blocks in $\mathcal P_n$ having side-length $3^n$ and the square containing $x$ denoted by $D_{n,x}$.  For each  $n, x$  we define the following as a \emph{large-field} event in $D_{n,x}$:
\be
\mathcal{F}_{n,x} \ :=  \ \{ \eta \, : \ \eps\,  \big|\eta(D_{n,x}) \big|   \ > \  J \, |\partial_{\text{e}} D_{n,x}| \} \,.
\ee
where $\eta(D) :=  \sum_{x\in D} \eta_x $ is the total block field.

A relevant feature of two dimensions is that the probabilities of the  large-field events are  scale invariant:
\be\label{eq:p_Mandelbrot_def}
\P(\mathcal F_{n,x}) \  = \  \chi(4J/\eps)  \ := p \  \approx \exp[-8J^2/\eps^2] \, .
\ee
For a given $x\in \Z^2$ the  events $\mathcal{F}_{n,x} $ are not strictly independent, however the sequence (in $n$) of the corresponding indicator functions is easily seen to be  asymptotic, in probability, to a stationary and mixing sequence of random variables.

Let $n(x; \eta)$ be the first non-negative integer for  which large field is exhibited in  $D_{n,x}$.
Due to the above properties of the events $\mathcal{F}_{n,x} $  for any $J,\eps>0$  almost surely  $n(x; \eta) < \infty$  for all $x$.

Under  $\mathcal F_{0,x}$, i.e. in case the large-field event occurs at $x$  already on the smallest scale, the value of the ground-state configuration at $x$ is predictably given by $\rm{sign} (\eta_x)$, i.e.  the sign of the field.   In case the $\eta_x$ is itself not large enough to meet this criterion, but the site is separated from the boundary of a set $\Lambda$ by a loop of sites  for which the large-field events occur at scale $n=0$, one may still conclude that  the finite-volume ground state at $x$ does not depend on the boundary spin configuration $\sigma_{\extB \Lambda}$.

Scaling up these observations, though along the way departing from rigor, we arrive at the following somewhat over-simplified algorithm for the assignment of a  spin configuration $\tau(\eta)$ which may mimic the infinite-volume ground state $\widehat\sigma(\eta)$.

For each $x\in \Z^2$ let $k(x; \eta) $ be defined as the smallest $0\le k< n(x; \eta)$ for which $x$ is separated from infinity by a loop of sites with $n(x; \eta) \leq k$, if such a $k$ exists, and otherwise set $k(x; \eta)=n(x; \eta)$.

In the first case, i.e. $n(x; \eta) =   k(x;\eta)$, we let   $\tau(\eta)_x = \rm{sign} (\eta(D(x))$.
If  $  k(x;\eta) < n(x; \eta) $, the value of   $\tau(\eta)_x$ is determined by minimizing the RFIM energy over the interior of the corresponding $x$-encapsulating loop, with the previously constructed values serving as boundary conditions for $\tau(\eta)_x$.

For the finite-volume version of the construction, in $\Lambda \subset \Z^2$, the above construction is modified by limiting the considerations of large-field events to cubes contained in $\Lambda$.  In the last step, unless $\tau^{\Lambda,\pm}(\eta)_x$ is defined  already through such events, its calculation will incorporate the boundary conditions imposed at $\extB \Lambda$.

Under the above algorithm the influence of the boundary conditions on $\tau(\eta)_x$ percolates over sites for which the events $\mathcal F_{n,x}$ did not yet occur.   For an idea on the  probability that the influence percolates deep inside $\Lambda$ one may take the further approximation in which the correlations between the indicator functions of nested events $\mathcal{F}_{n,x} $ are ignored.

Under the latter approximation, the collection of sites not covered by any of the  large-field events, has the distribution of the random fractal set discussed in Mandelbrot's ``canonical curdling'' model~\cite{M82}.
In particular,  the  influence-percolation process  coincides with the Mandelbrot-percolation process at density $p$ given by \eqref{eq:p_Mandelbrot_def}.

Curiously, as was proven by Chayes-Chayes-Durrett ~\cite{CCD88}, the Mandelbrot-percolation process does undergo a phase transition.   Its manifestation in the lattice version of the model is that the connectivity function decays exponentially fast for  $p$ large enough, but at $p$ small   the decay changes to a power law.  (The model is most appealing in its continuum, or ``ultraviolet'', limit while our discussion is focused on its infinite-volume, or ``infrared'', limit.  However in the analysis there is a simple relation  between the two).

It should however be noted that for  the finite-volume version of the construction, the existence of a path connecting $x$ to $\extB \Lambda$ in the complement of the set of sites covered by large-field events  is only a necessary condition for the dependence of $\tau^{\Lambda,\pm}(\eta)_x$ on the boundary conditions.
As its value   is determined through the energy minimization conditioned on  both the $\pm$ boundary conditions and the randomly determined values along the large-field sets, the $\pm$ boundary conditions may lose their effect on $\tau(\eta)_x$  even before the geometric disconnection of $x$ from $\extB \Lambda$.   Thus the Mandelbrot-percolation's  phase transition does not preclude exponential decay of the $\tau$-analog of our finite-volume order parameter at all $p>0$.

\medskip

\section*{Acknowledgements}
The work of MA was supported in part by the NSF grant DMS-1613296 and the Weston
Visiting Professorship at the Weizmann Institute. The work of RP was supported in part by Israel Science Foundation grant 861/15 and the European Research Council starting grant 678520 (LocalOrder).
We  thank the Faculty of Mathematics and Computer Science and the Faculty of Physics at
WIS for the hospitality enjoyed there during work on this project.

\end{document}